\documentclass[conference]{IEEEtran}
\ifCLASSINFOpdf
\else
\fi
\usepackage{graphicx}
\usepackage{array}
\usepackage{amsmath}
\usepackage{float}
\usepackage{float}
\usepackage{multirow}
\usepackage{array}
\usepackage{cancel}
\usepackage[normalem]{ulem}
\usepackage{tabu}
\usepackage{amsthm}
\usepackage{epsfig}
\usepackage{epstopdf}
\usepackage{epsf}
\usepackage{color}
\usepackage[english]{babel}
\usepackage{amsthm}

\usepackage{amsmath}
\usepackage{algorithm}
\usepackage[noend]{algpseudocode}

\newtheorem{definition}{Definition}
\newtheorem{proposition}{Proposition}

\theoremstyle{definition}
\newtheorem{exmp}{Example}[section]

\begin{document}
\medskip
%
\title{Achievable Rate Regions Using Novel Location Assisted Coding (LAC)}

\author{\IEEEauthorblockN{Thuan Nguyen}
\IEEEauthorblockA{School of Electrical and\\Computer Engineering\\
Oregon State University\\
Corvallis, OR, 97331\\
Email: nguyeth9@oregonstate.edu}
\and
\IEEEauthorblockN{Duong Nguyen-Huu}
\IEEEauthorblockA{School of Electrical and\\Computer Engineering\\
Oregon State University\\
Corvallis, OR, 97331\\
Email: nguyendu@eecs.oregonstate.edu}
\and
\IEEEauthorblockN{Thinh Nguyen}
\IEEEauthorblockA{School of Electrical and\\Computer Engineering\\
Oregon State University\\
Corvallis, 97331 \\
Email: thinhq@eecs.oregonstate.edu}
}


%


\maketitle

\begin{abstract}
The recent increase in number of wireless devices has been driven by the growing markets of smart homes and the Internet of Things (IoT).  As a result, expanding and/or efficient utilization of the radio frequency (RF) spectrum is critical to accommodate such an increase in wireless bandwidth.  Alternatively, recent free-space optical (FSO) communication technologies have demonstrated the feasibility of building WiFO, a high capacity indoor wireless network using the femtocell architecture.  Since FSO transmission does not interfere with the RF signals, such a system can be integrated with the current WiFi systems to provide orders of magnitude improvement in bandwidth.  A novel component of WiFO is its ability to jointly encode bits from different flows for optimal transmissions.  In this paper, we introduce the WiFO architecture and a novel cooperative transmission framework using location assisted coding (LAC) technique to increase the overall wireless capacity.  Specifically, achievable rate regions for WiFO using LAC will be characterized.  Both numerical and theoretical analyses are given to validate the proposed coding schemes.

\end{abstract}

Keyword: wireless, free space optical, capacity, achievable rate.

%
\IEEEpeerreviewmaketitle

\section{Introduction}
The number of wireless devices are projected to continue to grow significantly in the near future, fueled by the emerging markets for smart homes and the Internet of Things (IoT).
However, such an increase is anticipated to be hindered by the limited radio frequency (RF) spectrum.  Consequently, much research have been focused on utilizing the RF spectrum more effectively.

One promising approach is termed dynamic spectrum access (DSA).  Using DSA, the RF spectrum is allocated dynamically on both spatial and temporal dimensions.  For the DSA approach to work well,  many technical challenges must be overcome.  These include circuitry and algorithms for Cognitive Radio (CR) devices capable of sensing, sending, and receiving data on different RF bands.

Another notable approach uses femtocell architecture \cite{chandrasekhar2008femtocell}. Femtocell architecture has attracted enormous interest in recent years because its transmission ranges are limited within small cells, resulting in reduced interference and increased spectral efficiency.  That said, typical RF femtocells do not have large bandwidth to support a large number of users.  Alternatively, millimeter wave femtocells can be used to increase the bandwidth.  However, to achieve a large bandwidth, highly complex modulators and demodulators must be used, resulting in large energy consumption per bit.  On the other hand, recent advances in Free Space Optical (FSO) technology promise a complementary approach to increase wireless capacity with minimal changes to the existing wireless technologies.  The solid state light sources such as Lighting Emitting Diode (LED) and Vertical Cavity Surface Emitting Laser (VCSEL) are now sufficiently mature that it is possible to transmit data at high bit rates reliably with low energy consumption using simple modulation schemes such as On-Off Keying. Importantly, the FSO technologies do not interfere with the RF transmissions.  However, such high data rates are currently achievable only with point-to-point transmissions and not well integrated with existing WiFi systems.  This drawback severely limits the mobility of the free space optical wireless devices.

In \cite{wang2014channel}\cite{Liverman:16}, the authors proposed an indoor WiFi-FSO hybrid communication system called WiFO that promises to provide orders of magnitude improvement in bandwidth while maintaining the mobility of the existing WiFi systems. A video demonstration of WiFO can be seen at http:\slash \slash www.eecs.oregonstate.edu\slash $\sim$thinhq\slash WiFO.html.  WiFO aims to alleviate the bandwidth overload problem often associated with existing WiFi systems at crowded places such as airport terminals or conference venues.  WiFO modulates invisible LED light to transmit data in localized light cones to achieve high bit rate with minimal interference.
%

That said, in this paper, our contributions include:  (1) a novel channel model for short range FSO transmissions using Pulse amplitude modulation (PAM); (2) a novel cooperative transmission scheme, also known as location assisted coding (LAC) scheme that takes advantage of the receiver's location information to achieve high bit rates; (3) characterization of the multi-user achievable rate regions for the proposed channel using the proposed LAC.


\section{Related Work}
\label{sec:related}
From the FSO communication perspective, WiFO is related to several studies on FSO/RF hybrid systems. The majority of these studies, however are
in the context of outdoor point-to-point FSO transmission, using a powerful modulated laser beam.
There are also recent literature on joint optimization of simultaneous transmissions on RF and
FSO channels.   To obtain high bit rates  and spectral efficiency,  many FSO communication systems \cite{Haas:2016} use sophisticated modulation schemes such as Phase-Shift Keying (PSK) or Quadrature Phase-Shift Keying (QPSK) \cite{patnaik2013design} \cite{cvijetic200810gb} or Quadrature Amplitude Modulation (QAM)  \cite{kadhim201416} \cite{djordjevic2006multilevel} or Pulse Position Modulation (PPM)  \cite{nguyen2010coded} \cite{lee2011channel} \cite{sevincer2013lightnets}. However, these modulation schemes pay high costs in power consumption,  complexity,  and additional sensitivity to phase distortions of the received beam \cite{henniger2010introduction}.  In contrast, taking the advantage of high modulation bandwidth of recent LED/VCSEL and short-range indoor transmissions, WiFO uses simple Pulse Amplitude Modulation (PAM) \cite{Barry:2003:DCT:996027}, specifically ON-OFF Keying which results in simplicity and low power consumption.

From the coding's perspective,  the proposed LAC technique in WiFO is similar to MIMO systems that have been used widely in communication systems to improve the capacity \cite{goldsmith2003capacity}\cite{caire2003achievable}\cite{shen2007sum}.  Both LAC and MIMO techniques use several transmitters to transmit signals to achieve higher capacity. However, using multiple transmitters at the same time can also cause interference among transmissions to different receivers if they are in the same transmission range.  As such, a MIMO receiver typically receives signals from multiple transmit antennas and these signals are intended for that particular MIMO receiver at any time slot.  On the other hand,  in WiFO,  multiple transmitters transmit the joint messages simultaneously to multiple WiFO receivers, rather than a single receiver.   By taking advantage of the known interference patterns using the receiver location information, LAC technique can help the WiFO receivers to decode each message independently in presence of interference.  In a certain sense, this work is similar to the work of \cite{katti2006xors}. We note that a special case of LAC  technique was first introduced in \cite{duong2015location}. In this paper, we extend and improve the LAC  technique to obtain higher rates.

We note that our problem of characterizing the achievable region appears to be similar to the well-known broadcast channels \cite{cover1972broadcast}\cite{Marton:2006:CTD:2263338.2269512} .   Specifically, when the channel is a Degraded Broadcast Channel (DBC), the capacity region has been established \cite{cover1972broadcast}\cite{neuhoff1975coding} \cite{gallager1974capacity}. However, we can show that WiFO channel is not a degraded broadcast channel, thus  the well-known results  on DBC are not applicable \cite{technical_report_thuan_nguyen}.  Our work is also related to many research on multi-user MIMO capacity. These studies, however deal with Gaussian channels, rather than the specific proposed channel for WiFO systems.  In addition, while there have been many studies on the capacity of FSO channels \cite{shapiro2005ultimate} \cite{lapidoth2009capacity}, their focuses are mainly on modeling the underlying physics, and multi-user capacity is not considered. In contrast, we propose a simple channel model that lead to constructive coding schemes with corresponding achievable rate region for multi-user scenarios.

\section{Overview of WiFO Architecture}
\label{sec:WiFO}

\begin{figure}
  \centering
  \includegraphics[width=2.7in]{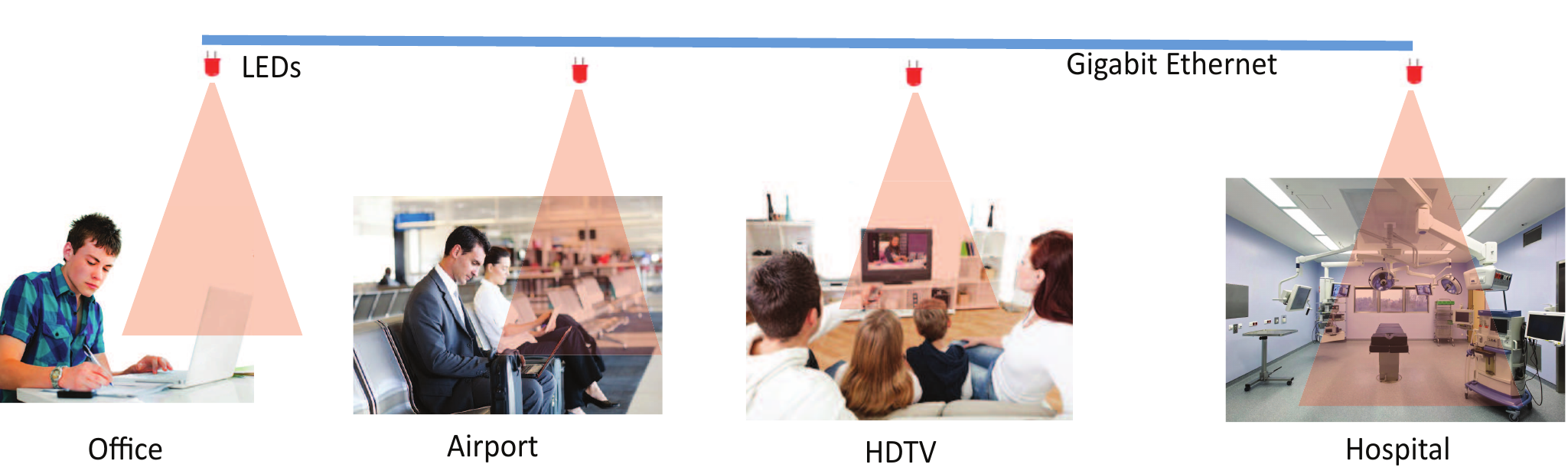}\\
  \caption{WiFO use scenarios}\label{fig:use}
  \vspace{-.2in}
\end{figure}

WiFO consists of an array of FSO transmitters to be deployed directly under the ceiling.  These FSO transmitters use inexpensive LEDs to modulate light via Pulse Amplitude Modulation (PAM). Fig.~\ref{fig:use} shows a few use cases for WiFO to boost up the wireless bandwidth.  These deployments include airport terminals, offices, entertainment centers, and automated device-device communications in critical infrastructures such as hospitals where cable deployment is costly or unsafe.

To transmit data, each FSO transmitter creates an invisible light cone about one square meter directly below in which the data can be received.  Fig.~\ref{fig:config}(a) shows a typical coverage area of WiFO using several FSO transmitters.  Digital bits ``1" and ``0" are transmitted by switching the LEDs on and off rapidly.  For the general PAM scheme, signals of more than two levels can be transmitted by varying the LED intensities.  The switching rate of the current system can be up to 100 MHz for LED-based transmitters and $>$ 1 GHz for VCSEL-based transmitters.  We note that, a number of existing FSO systems use visible light communication (VLC) which limits the modulating rate of a transmitter.  Thus, to achieve high bit rates, these systems use highly complex demodulators and modulators  (e.g. 64-QAM, OFDM), which make them less energy efficient.

Fig.~\ref{fig:config}(b) shows the light intensity as the function of the position measured from the center of the cone.  High intensity results in more reliable transmissions.

\begin{figure}
  \centering
  $\begin{array}{cc}
  \includegraphics[width=1.5in]{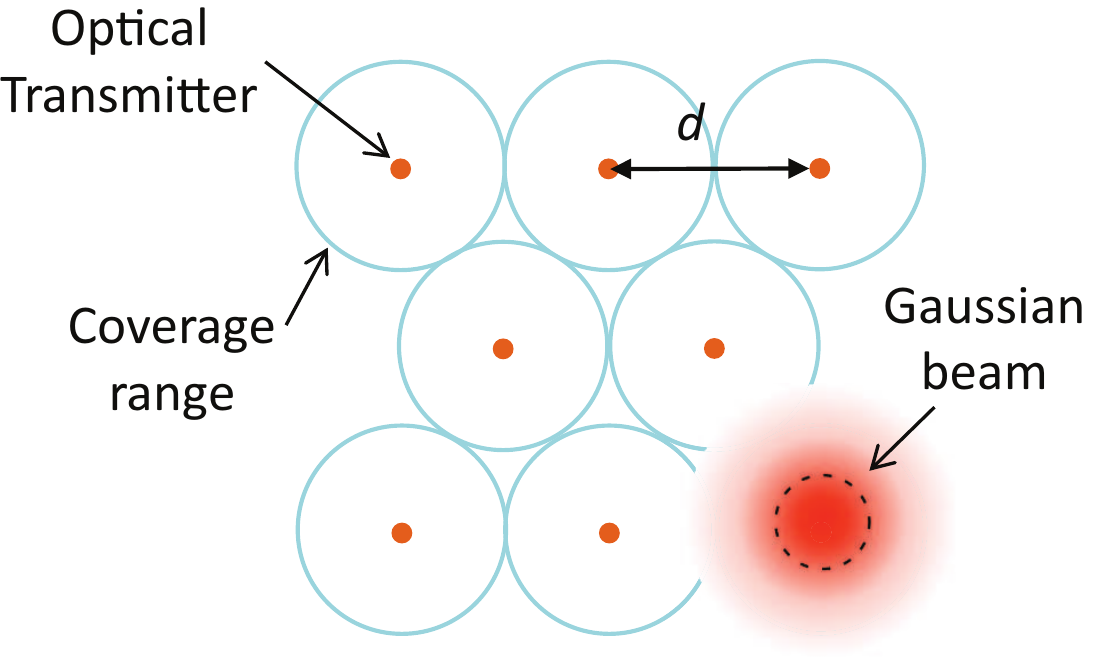} & \includegraphics[width=1.5in]{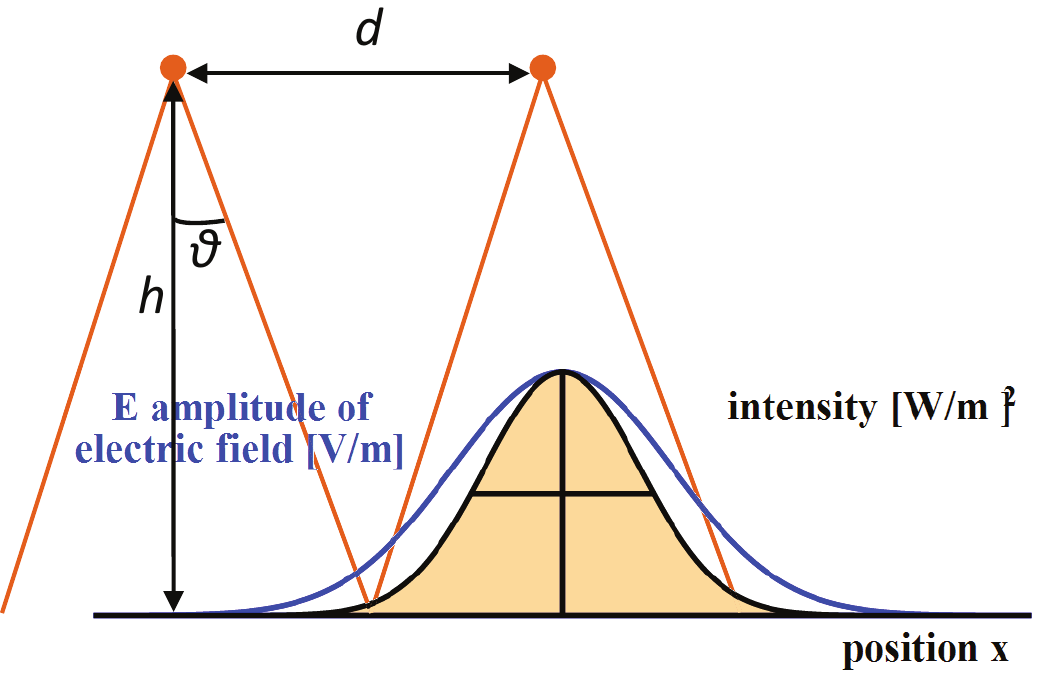} \\
  (a) &(b)
  \end{array}$
  \caption{(a) Configuration of the optical transmitter array; (b) coverage of optical transmitters with a divergent angle of $\vartheta$ }\label{fig:config}
\end{figure}

All the FSO transmitters are connected to a 100 Gbps Ethernet network which is controlled by the Access Point (AP). The AP is the brain of the WiFO system that controls the simultaneous data transmissions of each FSO transmitter and the existing WiFi channel.
At the receiving side, each WiFO receiver is equipped with a silicon pin photodiode which converts light intensity into electrical currents that can be interpreted as the digital bits ``0" and ``1".   The AP decides whether to send a packet on the WiFi or FSO channels. If it decides to send the data on the FSO channel for a particular device, the data will be encoded appropriately, and broadcast on the Ethernet network with the appropriate information to allow the right device to transmit the data. Upon receiving the data, the FSO transmitter relays the data to the intended device. Fig.~\ref{fig:data_flow} shows more detail on how data is transmitted from the Internet to the AP, then to the WiFO receiver over a FSO channel. Upon receiving the data from the FSO channel, the receiver decodes the data, and sends an ACK message to the AP via the WiFi channel.
ACK messages allow the system to adapt effectively to the current network conditions.
If the AP decides to send the data on the WiFi channel, then it just directly broadcasts the data through the usual WiFi protocol.

\begin{figure}
  \centering
  \includegraphics[width=2.5in]{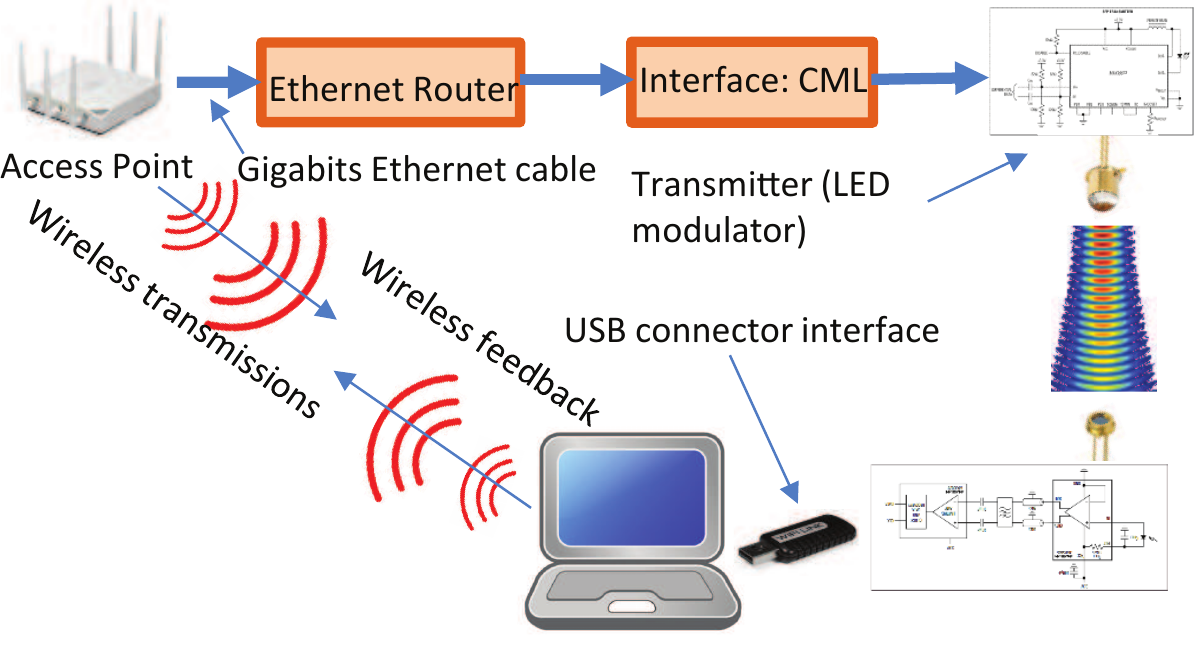}
  \caption{Data flow in WifO; Downlink connection uses both WiFi and FSO while uplink connection and ACKs use WiFi channel.}\label{fig:data_flow}
\end{figure}

As a receiver moves from one light cone to another, the AP automatically detects its location, and selects the appropriate LED to transmit the data.  The detection and selection of transmitters are performed quickly to prevent interruptions in data transmission.  Furthermore, even when the FSO transmitters are sparsely populated such that a user is not covered by any FSO transmitter,  all the data will be automatically sent via the existing WiFi channel.

One salient feature of WiFO is that, in a dense deployment scenario where light cones from LEDs are overlapped, a single receiver can associate with multiple LEDs.  As will be shown in Section \ref{sec:problem}, using cooperative transmissions from these LEDs via a novel location assisted coding (LAC) technique, a receiver in an overlapped area can receive higher bit rates.

\section{Problem Description}
\label{sec:problem}
In this section, we first provide some of the basic assumptions on the capabilities of WiFO.

\subsection{Assumption}
\label{sec:cooperative}

{\bf Location Knowledge.}
Because FSO transmitters are connected through a 100 Gbps Ethernet, the smart AP can control the transmission of individual FSO transmitter.  Furthermore, the AP knows the locations of all the receivers.  In particular, the AP knows which light cone that a receiver is currently located in.  This is accomplished through the WiFO's mobility protocol that can be described as follows.

Each FSO transmitter broadcasts a beacon signal consisting of a unique ID periodically. Based on its location, a receiver will automatically associate with one or more transmitters that provide sufficiently high SNR beacon signals.  Upon receiving the beacon signal from a transmitter, the receiver sends back {\em alive} heartbeat messages that include the essential information such as
the transmitter ID and the  MAC addresses to the AP using WiFi channel.  The AP then updates a table whose entries consist of the MAC address and the transmitter IDs which are used to forward the packets of a receiver to the appropriate transmitters.  If the AP did not receive a heartbeat from a device for some period of time, it will disassociate that device, i.e., remove its MAC address from the table.   Thus, the location information of a receiver is registered automatically at the AP.

{\bf Sparse vs. Dense Deployment.} Sparse deployment of FSO transmitters leads to less FSO coverage, but is resource efficient.  On the other hand, a dense deployment increases mobility and the bit rates for a single receiver if two or more transmitters are used to transmit data to a single receiver.  However, a dense deployment also leads to multi-user interference that might reduce the overall rate.  In this paper, we are interested in a dense deployment and show that the multi-user interference is not necessary when the side information, specifically the knowledge of receiver locations is incorporated into the proposed cooperative transmission scheme or LAC technique.

{\bf Transmitter.}  We assume that there are $n$ FSO transmitters $T_1, T_2, \dots T_n$,  each produces a light cone that overlaps each other. There are also $m$ receivers denoted as $R_1, R_2, \dots R_m$.
A FSO transmitter is assumed to use PAM for transmitting data.  However, to simplify our discussion, we will assume that a sender uses On-Off Keying (OOK) modulation where high power signal represents ``1'' and low power signal represents ``0'' \cite{henniger2010introduction}.  We note that the proposed LAC scheme can be easily extended to work with the general PAM.

{\bf Receiver.} A receiver is assumed to be able to detect different levels of light intensities. If two transmitters send a ``1'' simultaneously to a receiver, the receiver would be able to detect ``2'' as light intensities from two transmitters add constructively. On the other hand, if one transmitter sends a ``1'' while the other sends a ``0'', the receiver would receive a ``1''.

\subsection{Channel Model}
\label{sec:channel}

To assist the discussion, we start with a simple topology consisting of transmitters and two receivers shown in Fig. \ref{fig:n_2_m_2}(a). Receiver $R_2$ is in the overlapped area, and therefore can receive the signals from both transmitters while receiver $R_1$ can receive signal from only one transmitter.
\begin{figure}[h]
\begin{center}
$\begin{array}{cc}
\includegraphics[scale=0.25]{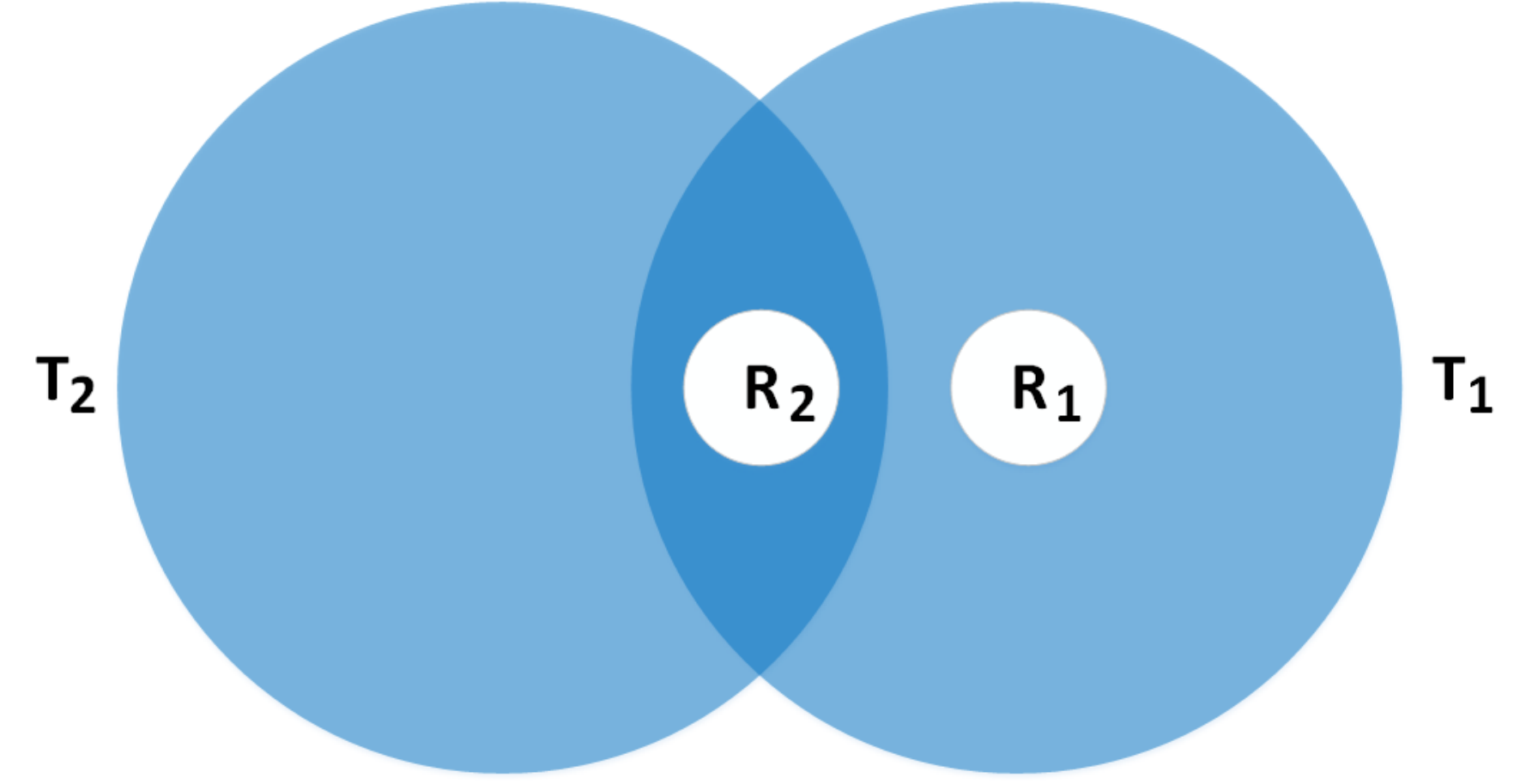} &
\includegraphics[scale=0.12]{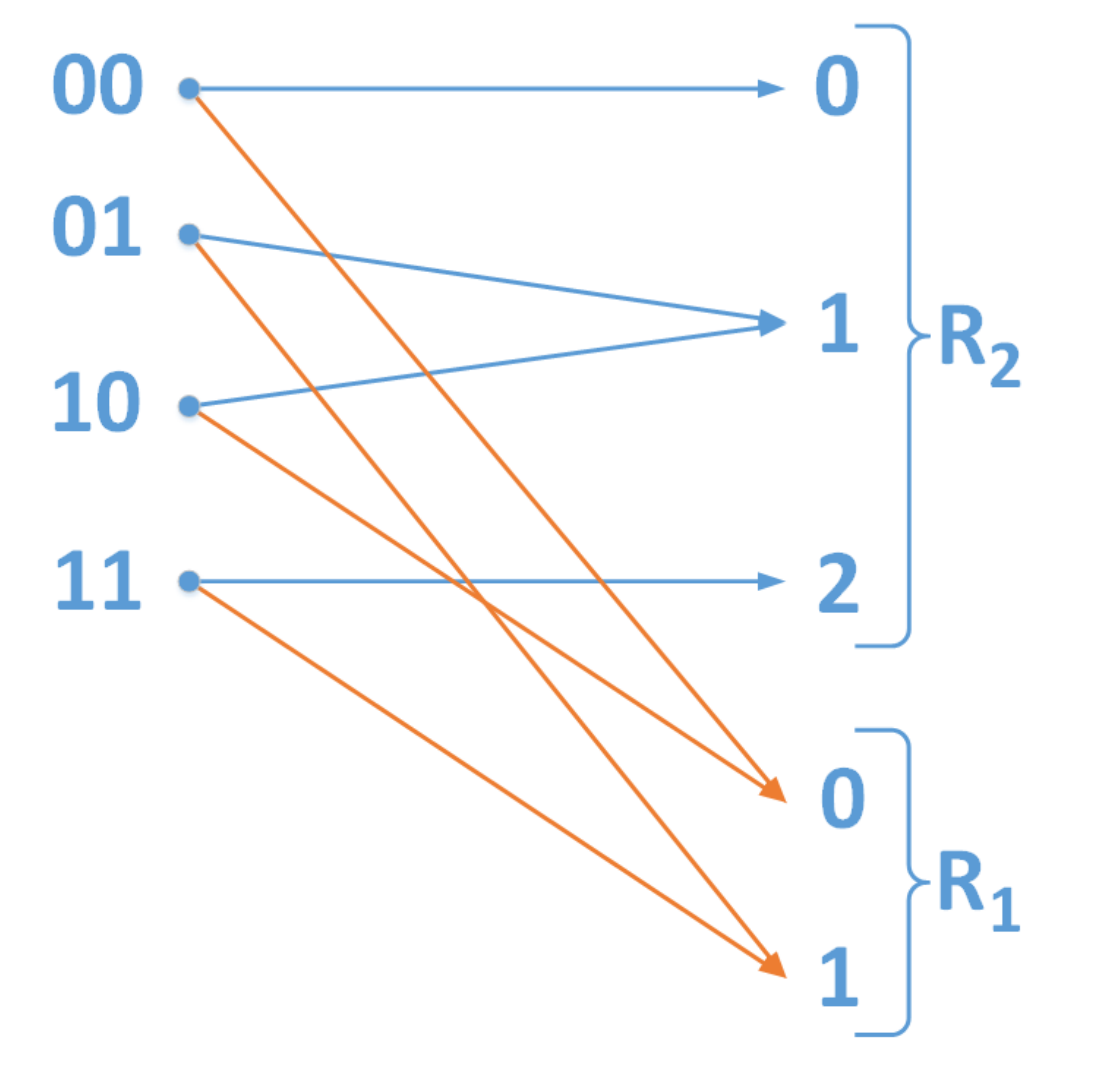} \\
(a) & (b)
\end{array}$
\end{center}
\caption{(a) Topology for two transmitters and two receivers; (b) Broadcast channels for two receivers.}
\label{fig:n_2_m_2}
\end{figure}
Cooperative transmission scheme uses both transmitters to send independent information to each receiver simultaneously.  This cooperative transmission scheme can be viewed as a broadcast channel in which the sender can broadcast four possible symbols: ``00", ``01", ``10", and ``11" with the left and right bits are transmitted by different transmitters.  Thus, there is a different channel associated with each receiver.  Fig. \ref{fig:n_2_m_2}(b) shows the broadcast channel for the two receivers $R_1$ and $R_2$.   There are only three possible symbols for  $R_2$ because it is located in the overlapped coverage of two transmitters.  Therefore, it cannot differentiate the transmitted patterns ``01" and ``10"  as both transmitted patterns result in a ``1" at $R_2$  due to the additive interference.  On the other hand, there are only two symbols at receiver $R_1$ because it is located in the coverage of a single transmitter.

Similarly, Fig. \ref{fig:n_3_m_2}(a) shows a topology with three transmitters and two receivers and Fig. \ref{fig:n_3_m_2}(b) shows the corresponding broadcast channels.

\begin{figure}[h]
\begin{center}
$\begin{array}{cc}
\includegraphics[scale=0.12]{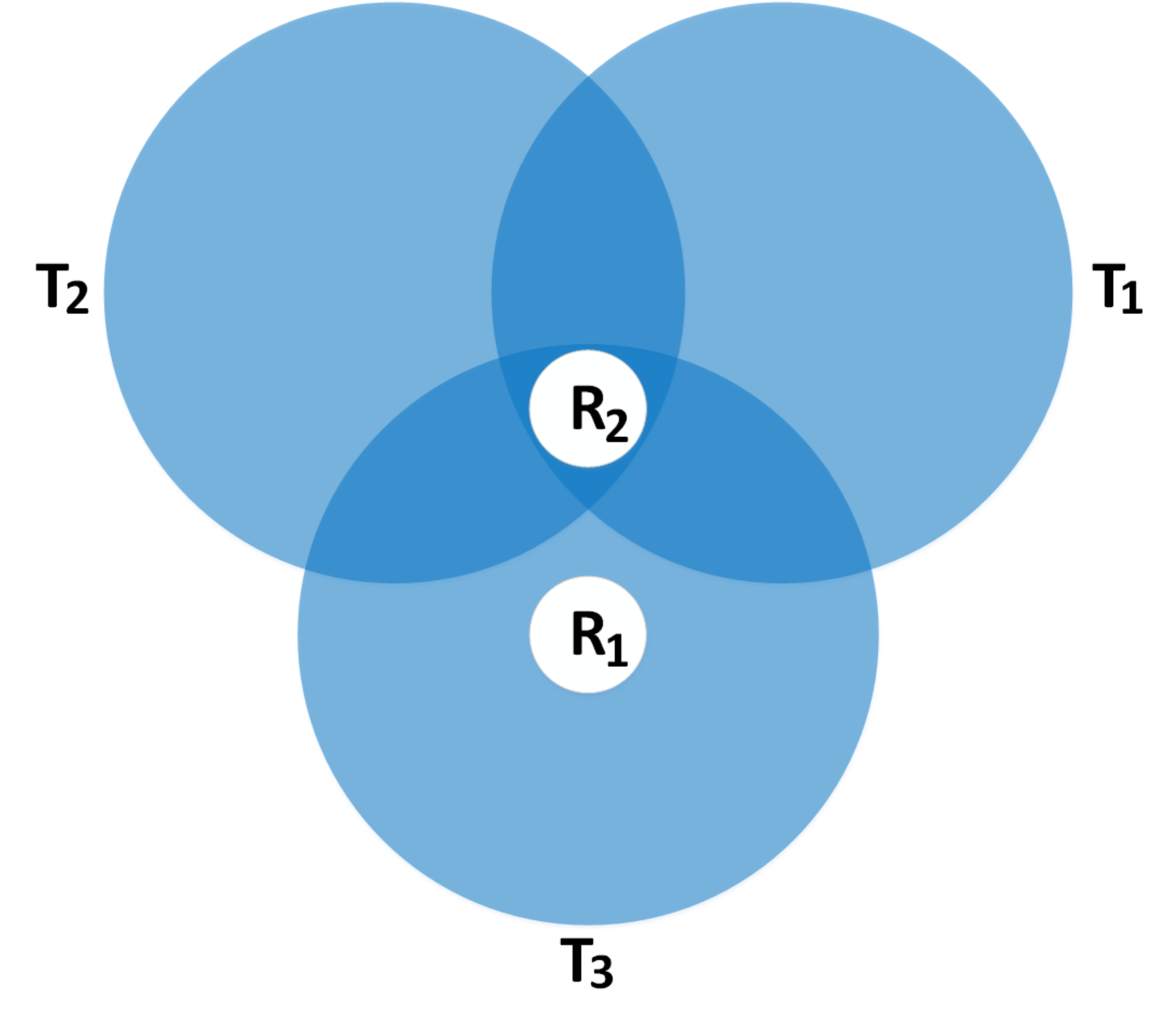} &
\includegraphics[scale=0.12]{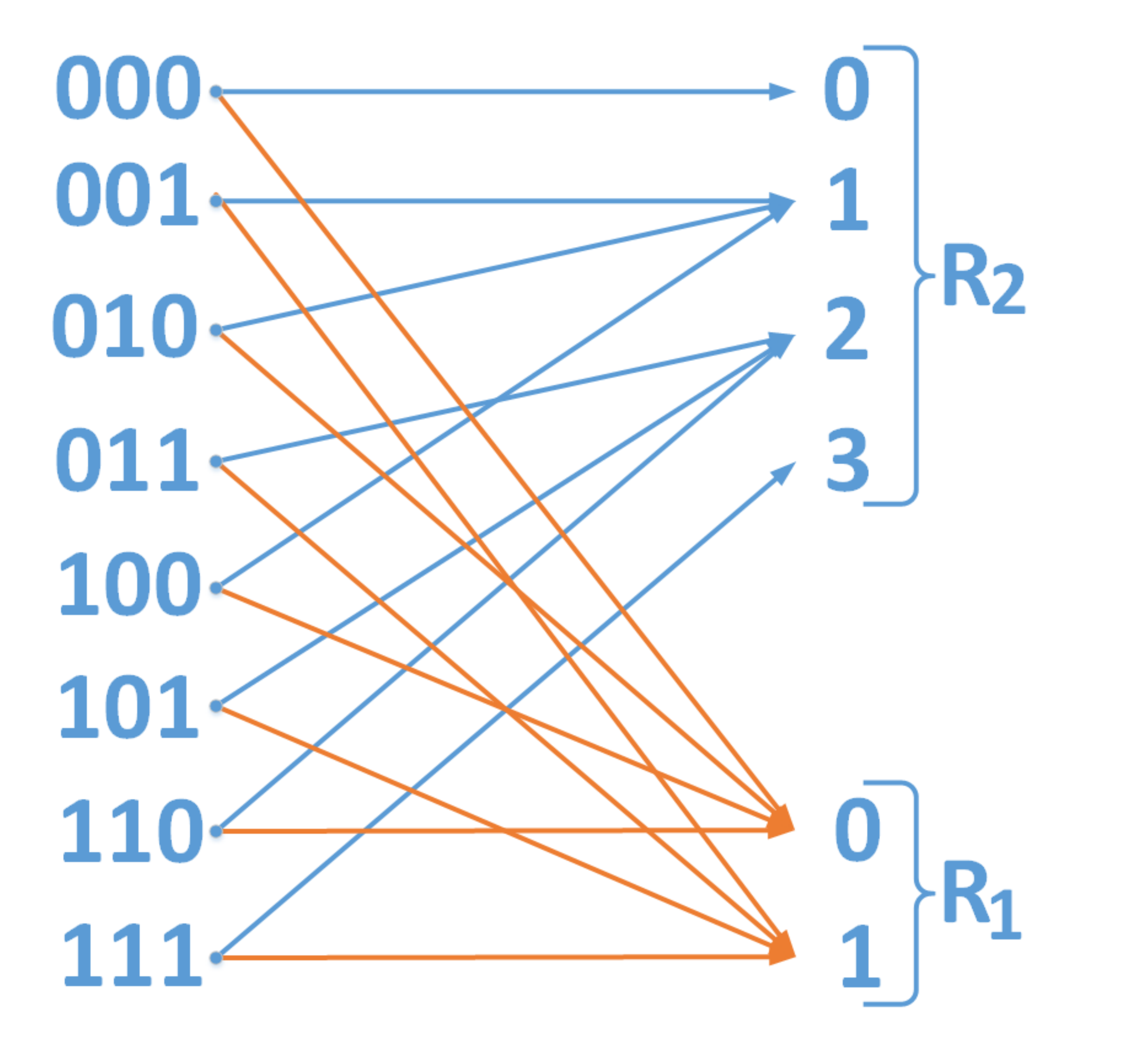}\\
(a) & (b)
\end{array}$
\end{center}
\caption{(a) Topology for three FSO transmitters and two receivers; (b) Broadcast channels for two receivers.}
\label{fig:n_3_m_2}
\end{figure}

Assuming that there is no transmission errors, then it is straightforward to see that the channel matrices for $R_1$ and $R_2$ associated with Fig. \ref{fig:n_2_m_2}(b) are:
\\
$\begin{array}{cc}
A_1 = \begin{bmatrix}
1 & 0  \\
0 & 1 \\
1 & 0 \\
0 & 1 \\
\end{bmatrix}, &
A_2 = \begin{bmatrix}
1 & 0 & 0 \\
0 & 1 & 0\\
0 & 1 & 0\\
0 & 0 & 1\\
\end{bmatrix}
\end{array}. $\\

We note that the entry $A(i,j)$ of the channel matrix denotes probability that a transmitted symbol $i$ to turn a symbol $j$ at the receiver.  Since we assume all sources of error are due to multi-user interference,
$A(i,j)$ is either 0 or 1.

Similarly, the channel matrices for $R_1$ and $R_3$ associated with Fig. \ref{fig:n_3_m_2}(b) are:\\

$\begin{array}{cc}
A_1 = \begin{bmatrix}
1 & 0  \\
0 & 1 \\
1 & 0 \\
0 & 1 \\
1 & 0  \\
0 & 1 \\
1 & 0 \\
0 & 1 \\
\end{bmatrix} , &
A_3 = \begin{bmatrix}
1 & 0 & 0 & 0 \\
0 & 1 & 0 & 0 \\
0 & 1 & 0 & 0 \\
0 & 0 & 1 & 0 \\
0 & 1 & 0 & 0 \\
0 & 0 & 1 & 0 \\
0 & 0 & 1 & 0 \\
0 & 0 & 0 & 1 \\
\end{bmatrix}
\end{array}.$\\

The same method  can be used to construct the channel matrices for arbitrary configurations/topologies with  different numbers of transmitters and receivers.
For clarity, in this paper, we only discuss the coding techniques and achievable capacity region for ideal channels with no errors.  However, the proposed techniques can be readily extended to channel with errors by
constructing a different channel matrix.

\subsection{Achievable Rate Region}
Achievable rate region characterizes  the rates at which each receiver can receive their independent information simultaneously.
{\em Our goal is to determine a cooperative transmission scheme among the transmitters in order to enlarge the achievable rate region for the receivers.}

To discuss the achievable rate region,  we use an example given by the topology shown in  Fig. \ref{fig:n_2_m_2}(a).  We assume that the transmitters $T_1$ and $T_2$ are responsible for transmitting the independent information to its receivers $R_1$ and $R_2$ respectively.  Suppose $R_1$ and $R_2$ want to receive bits "1" and "0", respectively.   If  $T_1$ and $T_2$ can naively transmit bit "1" and "0", respectively, then $R_1$ will correctly receive its bit "1".   On the other hand,  since $R_2$ is located in the overlapped coverage of the two transmitters, it will incorrectly receive bit ``1" due to the additive multi-user interference.  To resolve the multi-user interference,  a TDMA scheme can be employed in which each  transmitter can take turn to transmit a  bit to its receiver in each time slot.   As a result, using the naive scheme coupled with TDMA, on average  each receiver can receive 0.5 bit per time slot.  Another scheme would be just to transmit bits to either $R_1$ or $R_2$ exclusively.  This implies that one receivers will have 1 bit per time slot while the other zero bit per time slot.  Thus, let ($x$, $y$) denote the achievable rate tuple where $x$  and $y$ denote the average of $R_1$ and $R_2$, then achievable rate region would include the rate tuples:  (1,0), (0,1), (0.5,0.5).  In general, a time-sharing strategy that uses  the scheme (1,0) for $\lambda$ fraction of the time, and the scheme (0,1) for $1- \lambda$ of the time  produces a rate region shown in Fig. \ref{fig:simple_rate_region_2_2}.  In Section \ref{sec:LAC}, we will show that  such a scheme produces a suboptimal (small) rate region, and describe how the LAC technique can be used to enlarge the achievable rate region.
\begin{figure}
  \centering
  \includegraphics[width=2in]{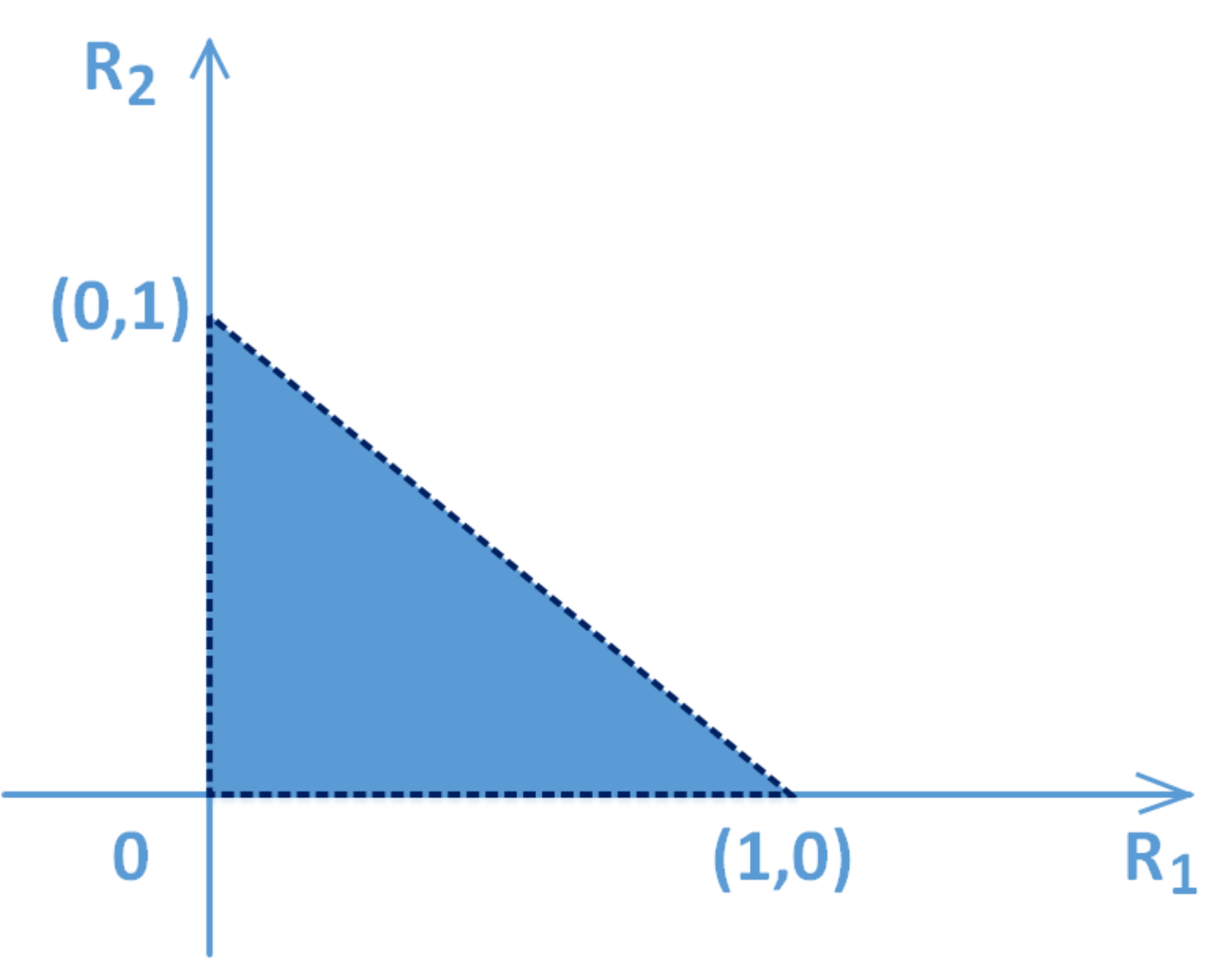}
  \caption{Achievable rate region using time-sharing strategy between two tuples (0,1) and (1,0)}\label{fig:simple_rate_region_2_2}
  \vspace{-.2in}
\end{figure}

Also, we note that the proposed cooperative transmission scheme/coding technique can be extended to handle the channels with external errors.

\section{Cooperative Transmission via Location Assisted Coding (LAC)}
\label{sec:LAC}
LAC is a cooperative transmission scheme that uses the receiver's location information to enlarge the achievable rate region.    For a given topology, LAC employs different coding schemes: single rate coding (SRC), equal rate coding (ERC),  and joint rate coding (JRC).  Each scheme finds a different feasible rate tuple.  Next, by varying the fractions of the time that LAC uses these different coding schemes,  the achievable rate region can be achieved as the convex hull of these rate tuples.

\subsection{Single Rate Coding}
\label{sec:single}
Using SRC, a receiver in the coverage of $n$ transmitters, can receive high bit rate by using all $n$ transmitters to transmit the information for that particular receiver. As a result, other receivers even though located in the coverage of some of these $n$ transmitters, will not receive any information.  We have the following results on the achievable rate of the single receiver.

\begin{proposition} (Single Rate Coding)
\label{prop:SRC}
For a receiver in the light cone of $n$ transmitters, the achievable rate is $\log{(n + 1)}$ bits per time slot.
\end{proposition}

\begin{proof}
Since each transmitter is capable of transmitting ``0" or ``1" only, and the single receiver receives the sum of all the signals from the $n$ transmitters, then there is total of $n+1$ distinct levels perceived at the receiver.    Furthermore, since there is no error involved, the probability mass function of the transmitted symbols is identical of the probability mass function of the received symbols. Thus, from basic result of information theory \cite{Cover:2006:EIT:1146355}, the capacity for the single user is achieved using the uniform probability mass function which results in $\log{(n+1)}$ bits per time slot.  Note that the rates of other receivers is zero.
\end{proof}

\subsection{Equal Rate Coding}
\label{sec:proportional}
In SRC, one receiver receives high bit rate while rates for other receivers are zero.   On the other hand, using ERC , for certain topologies,  each receiver to obtain one independent bit per time slot.   Let  $H$ to be the topological matrix whose entry $H(i,j)$ is equal to $1$ if receiver $i$ can receive signal from transmitter $j$ and $0$ otherwise. For example, the topological matrix associated with Fig. \ref{fig:n_2_m_2}(a) is:
$$ H = \begin{bmatrix}
       1 & 0  \\
       1 & 1
     \end{bmatrix}.$$

Assume $H$ is full rank, and for simplicity, the number of receivers equal the number of transmitters, then the proposition for ERC are as follows.

\begin{proposition} (Equal Rate Coding \cite{duong2015location})
\label{prop:ERC}
If the topological matrix $H$ is full-rank, then using ERC, every receiver can receive 1 bit per time slot.  Furthermore, if $H$ is an $n \times n$ full rank matrix, then maximum sum of all receiver rates  is $n$ bits per time slot.
\end{proposition}
\begin{proof}

We will show explicitly the encoding and decoding procedures to obtain 1 bit per time slot for each receiver using ERC.\\
\textbf{Encoding}: Let $b = (b_1, b_2, \dots, b_n)^T$ denote the information bits intended to be sent to receiver $R_1, R_2, \dots, R_n$, respectively.  $x = (x_1, x_2, \dots, x_n)^T$ be the coded bits transmitted by the transmitter $T_1, T_2, \dots, T_n$, respectively, and $y = (y_1, y_2, \dots, y_n)^T$ be the signal received at the receiver $R_i$.  The goal of the encoding scheme   $x = \mathcal{C}(b)$,  is to produce the bits $x_i$'s such that every receiver $R_i$, upon receiving $y_i$, can recover its $b_i$.

We consider the following system of linear equations:
\begin{small}
\begin{equation}
\label{eq:encoding}
\begin{cases}
\begin{array}{r@{}l}
H(1,1) x_1 \oplus H(1,2) x_2 \oplus \ldots \oplus H(1,n) x_n &{}= b_1 \\
H(2,1) x_1 \oplus H(2,2) x_2 \oplus \ldots \oplus H(2,n) x_n &{}= b_2 \\
&{}\ldots  \\
H(n,1) x_1 \oplus H(n,2) x_2 \oplus \ldots \oplus H(n,n) x_n &{}= b_n
\end{array}
\end{cases}
\end{equation}
\end{small}
where $\oplus$ is addition in $\mathbf{GF}(2)$, i.e. $a \oplus b = (a+b)\mod2$.   Since $H$ is full-rank  in $\mathbf{GF}(2)$, we can solve the system of equations \eqref{eq:encoding} above for unique $x_1$, $x_2$, $\ldots$, $x_n$ in terms of $b_1$, $b_2$, $\ldots$, $b_n$.  Mathematically, the encoding is:

\begin{equation}
x = H^{-1}b,
\end{equation}
where all the computations are done in finite field $\mathbf{GF}(2)$.
Each transmitter $T_i$ then transmits $x_i$'s  to the receivers.

\textbf{Decoding}: A receiver $R_i$ needs to be able to recover the bit $b_i$ from the received signal $y_i$ which can be represented as:

\begin{small}
\begin{equation}
\label{eq:receive}
\begin{cases}
\begin{array}{r@{}l}
    y_1 &{}= H(1,1) x_1 + H(1,2) x_2 + \ldots + H(1,n) x_n\\
    y_2 &{}= H(2,1) x_1 + H(2,2) x_2 + \ldots + H(2,n) x_n\\
    &{} \ldots \\
    y_n &{}= H(n,1) x_1 + H(n,2) x_2 + \ldots + H(n,n) x_n
\end{array}
\end{cases}
\end{equation}
\end{small}

Note that the addition $+$ in \ref{eq:receive} is ordinary addition operation.

Now upon receiveing $y_i$'s, the receiver $R_i$ recovers $b_i$ by performing
\begin{equation}
\label{eq:recovered}
y_i \mod 2 = \hat{b}_i.
\end{equation}

It is easy to check that $b_i = \hat{b}_i$. This can be seen by performing $\mod 2$ operations on both sides of equations \eqref{eq:receive} which results in the equations \eqref{eq:encoding}.
Or simply, if $y_i$ is even then $R_i$ decodes bit $b_i$ as  ``0", and  ``1" otherwise.  As a result, each
receiver can decode its bits correctly and independently in presence of interference.

The second statement of the proof is straightforward.  We note that the sum rate is upper bounded by the maximum number of independent bits that can be sent out simulantanously.  Since there are $n$  transmitters, there are at most $n$ bits can be sent out simultaneously.  Since we showed that the for a full rank $H$, each receiver can receive 1 bit per time slot, and therefore when $H$ is  a $n \times n$ full rank matrix, the total rate is $n$ bits per time slot.

\end{proof}

Proposition \ref{prop:ERC} establishes the sufficient conditions regarding the topology that allows for (1) independent information to be sent at equal rates to all the receivers and (2) achieving maximum sum rate.

\subsection{Joint Rate Coding}
\label{sec:equal}

Unlike ERC, Joint Rate Coding (JRC) technique allows the receivers to obtain different rates.  JRC is a bit more involved.  To aid the discussion,  we employ the following definitions and notations.

%
%
%

\theoremstyle{definition}
\begin{definition}
\label{def1}
{(Exclusive and Shared Transmitters)}
Let $\mathcal{R} = \{1,2, \dots, m\}$ be the set of $m$ receivers.
Let $\mathcal{S}  \subset \mathcal{R}$,  and $\mathcal{T_S}$ denotes a group of transmitters that cover  exactly all the receivers in $S$.
Each transmitter in $\mathcal{T}_\mathcal{S}$ is called an exclusive transmitter if $\mathcal{S}$ is a singleton, and  a shared transmitter if $S$ has two or more elements.
Let $t_\mathcal{S} = |\mathcal{T}_\mathcal{S}|$ denote the number of transmitters that covers exactly all the receivers in $\mathcal{S}$. To simplify the notations,
for exclusive transmitters, we use $t_i$ to denote the number of transmitters that covers the receiver $R_i$ exclusively while $t_{ij}$ denotes the number of pairwise sharing transmitters that cover only two receivers $R_i$ and $R_j$ and no other receivers.
\end{definition}

For example,  in Fig. \ref{fig:n_2_m_2}(a), the transmitter $T_2$ is the only exclusive transmitter for $R_2$, and so $t_2 = 1$. On the other hand, $t_1$ = 0 since there is no exclusive transmitter for $R_1$.  However, $T_1$ is a shared transmitter between $R_1$ and $R_2$, so $t_{12}$ = 1.  Similarly, in Fig. \ref{fig:n_3_m_2} (a), $t_1 = 0$, $t_2 = 2$, and $t_{12} = 1$.

The key to the JRC technique is how to use the shared transmitters to transmit bits to multiple receivers simultaneously.  At the fundamental level, we develop JRC technique  for topologies that consist only exclusive and pairwise sharing transmitters.  Fig. \ref{fig:n_2_m_2}(a) and  \ref{fig:n_3_m_2} (a) show such topologies.  We then show how to decompose a general topologies into the several pairwise sharing topologies, then the fundamental techniques for pairwise can be applied. That said, we will first consider a two receivers $R_1$ and $R_2$ with $t_1$ and $t_2$ exclusive transmitters and $t_{12}$ shared transmitters.

JRC allocates different rates to the receivers $R_1$ and $R_2$ through two parameters, which can be viewed as the number of shared transmitters allocated to $R_1$ and $R_2$.  In particular, we denote
$t^1_{12}$  and $t^2_{12}$ as the number of shared transmitters allocated to $R_1$ and $R_2$, respectively.  We have:
 \begin{equation}
 \label{19}
 t^1_{12} +  t^2_{12} \leq t_{12}.
 \end{equation}

We will show using JRC, by increasing $t^1_{12}$, we allow $R_1$ to achieve higher rate at the expense of a reduced rate for $R_2$.   Fig. \ref{fig:jrc} illustrates our notations. Based on this, we have the following proposition on the achievable rates using JRC for two receivers.
\begin{figure}
  \centering
  \includegraphics[width=2.5in]{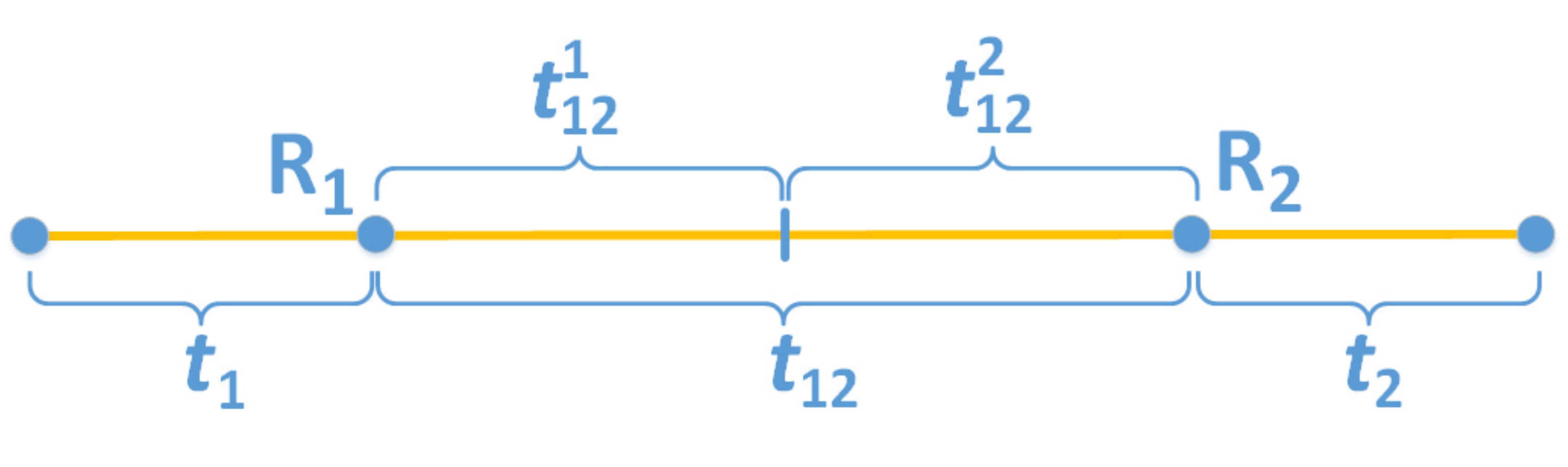}
  \caption{$t_1$ and $t_{2}$  are number of exclusive transmitters for $R_1$ and $R_2$ while $t_{12}=t_{21}$ is the number of transmitters that covers both $R_1$ and $R_2$; $t^{1}_{12}$ can be distributed to $R_1$ and $t^{2}_{12}$ can be distributed to $R_2$ to adjust the rates of $R_1$ and $R_2$.}\label{fig:jrc}
  \vspace{-.2in}
\end{figure}

\begin{proposition} {(Achievable rates for two-receiver topology)}.
\label{prop:JRC}
If  $t_1 \geq t^2_{12}$ and  $t_2 \geq t^1_{12}$  then $R_1$ and $R_2$ can achieve the rates of $\log{c_1}=\log{(t_{1} + t^1_{12} + 1)}$ and $\log{c_2}=\log{(t_{2} + t^2_{21} + 1)}$ bits per time slot, respectively, where $t^1_{12} + t^2_{12} \leq t_{12}$.  $t^1_{12}$ and  $t^2_{12}$ are parameters that control the rates between $R_1$ and $R_2$.
\end{proposition}
Note that to maximize the rates, we want $t^1_{12} + t^2_{12} = t_{12}$.
\begin{proof}

We will describe a constructive proof for Proposition \ref{prop:JRC}. But first,  let $x_{12}$ be a non-negative integer represented by the bit patterns sent out by $t_{12}$ shared transmitters.  Since each shared transmitter can send either a "0" or "1",  $x_{12}$ has $t_{12} + 1$ levels, i.e.,  $x_{12} \in \{0,1, \dots, t_{12}\}$.
Let $x_i$ be a non-negative integer that represents the bit patterns transmitted by $t_i$ exclusive transmitters for receiver $R_i$.  $x_i$ has $t_i$ + 1 levels, i.e., $x_i \in \{0, 1, \dots, t_i\}$.
Let $y_i$ be a non-negative integer that represents the signal received by the receiver $R_i$. Due to additive property, we have:
\begin{equation}
\label{eq:received}
y_i = x_i + x_{12}.
\end{equation}

Next, we note that the achievable rate of a receiver $R_i$ is $\log$ of the number of distinguishable symbols or levels that can be received by $R_i$ per time slot.
Let $c_i$ be a non-negative integer representing the number of distinguishable levels at $R_i$, then $\log{c_i}$ is the achievable rate of $R_i$.  We will show that if $t_1 \geq t^2_{12}$ and  $t_2 \geq t^1_{12}$,
then it is possible to send any arbitrary pattern pair $(b_1, b_2)$ to the receiver $R_1$ and $R_2$ without any error, with
$$b_i \in \{0, 1, \dots,  c_i-1\}.$$
This would establish the proof for Proposition \ref{prop:JRC}. We now describe the encoding and decoding procedures, then verify their correctness.

\textbf{Encoding}: Suppose we want to transmit the pattern $(b_1, b_2)$ to $(R_1,R_2)$, respectively. Then, the encoding is a function that maps $(b_1, b_2)$ into $x_1^*, x_2^*$, and $x_{12}^*$
,i.e.,  $(x_1^*, x_2^*, x_{12}^*) = \mathcal{C}(b_1, b_2)$.
Let the set $\{x_{12}(b_1)\}$ parameterized by $b_1$ consisting of $t_1+1$ elements be defined as:
\begin{equation}
\label{13}
\{x_{12}(b_1)\} = \{b_1-i_1 \mod (c_1), i_1 = 0, 1, \dots, t_1\}.
\end{equation}

Similarly, let the set  $\{x_{12}(b_2)\}$ parameterized by $b_2$ consisting of $t_2+1$ elements be defined as:
\begin{equation}
\label{14}
\{x_{12}(b_2)\} = \{b_2-i_2 \mod (c_2), i_2 = 0, 1, \dots, t_2\}.
\end{equation}

We now encode $b_1$, $b_2$ into  $x_1^*$, $x_2^*$, and $x_{12}^*$ as follows.  We pick $x^*_{12}$ to be the minimum value element in the intersection set of $\{x_{12}(b_1)\}$ and $\{x_{12}(b_2)\}$, i.e., :
$$ x^*_{12} = \min_i \{ x_i \in \{x_{12}(b_1)\} \cap \{x_{12}(b_2)\} \}.$$

Next, we set $x_i^*$, $i = 1,2$ to:
\begin{equation}
\label{15}
x_i^* = b_i - x^*_{12}  \mod (c_i).
\end{equation}

\textbf{Decoding}: $R_i$ receives the signal:
\begin{equation}
\label{eq:received_optimal}
y_i = x_i^* + x_{12}^*,
\end{equation}
the sum of the signals transmitted by the exclusive transmitters and shared transmitters.  $R_i$ decodes the transmitted level $b_i$ as:

\begin{equation}
\label{eq:decoding}
\hat{b}_i = y_i \mod (c_i).
\end{equation}


To verify the correctness of encoding and decoding procedures, we need to verify (a) $\{x_{12}(b_1)\} \cap \{x_{12}(b_2)\}$ is non-empty that enables us to choose $x*_{12} = \min {\{x_{12}(b_1)\} \cap \{x_{12}(b_2)\}}$; (b)  $x^*_{12} \le t_{12}$. This  is required since we want the $t_{12}$ shared transmitters to be able to represent $x^*_{12}$; (c) $0\le x_1^* \le t_1$ and $0 \le x_2^*  \le t_2$ to enable the exclusive transmitters to represent $x_i$; (d) $\hat{b}_i=b_i$ for the correctness of the decoding procedure.

First, we will verify the condition (a). From the definition (Eqs. (\ref{13}) and (\ref{14}), the sets
$\{x_{12}(b_i)\}$  consists of $(t_i+1)$ distinct elements each.  Furthermore,

$$\{x_{12}(b_i)\} \subseteq \{0,1,\dots , \max(c_1-1,c_2-1)\},$$
$$|\{x_{12}(b_1)\} \cup \{x_{12}(b_2)\}| \leq \max(c_1,c_2).$$

The number of elements in $\{x_{12}(b_1)\} \cap \{x_{12}(b_2)\}$ set is:
\begin{eqnarray}
|\{x_{12}(b_1)\} \cap \{x_{12}(b_2)\}| &=& |\{x_{12}(b_1)\}| + |\{x_{12}(b_2)\}| \nonumber\\
&-& |\{x_{12}(b_1)\} \cup \{x_{12}(b_2)\}| \nonumber \\
&\geq&  t_1+1+t_2+1-\max(c_1,c_2). \nonumber
\end{eqnarray}

Now since $c_1=t_1+t_{12}^1+1$ and $c_2=t_2+t_{12}^2+1$, we have:
\begin{small}
\begin{eqnarray}
\label{12}
|\{x_{12}(b_1)\} \!\cap \{x_{12}(b_2)\}| \,\geq  \min(t_2-t^1_{12}+1, t_1-t^2_{12}+1)
\end{eqnarray}
\end{small}
Using the conditions in Proposition \ref{prop:JRC}: $t_1 \geq t_{12}^{2}$ and $t_2 \geq t_{12}^{1}$, we conclude the intersection set $|\{x_{12}(b_1)\} \!\cap \{x_{12}(b_2)\}|$ has at least one element, and therefore we can pick $x^*_{12}$.

Next, we will prove condition (b) by contradiction by assuming
\begin{equation}
\label{11}
x^*_{12} > t_{12}.
\end{equation}

Let $x_{12}^{\max}$ be the maximum element in $\{x_{12}(b_1)\} \cap \{x_{12}(b_2)\}$. Then,
\begin{small}
\begin{eqnarray}
\! x_{12}^{\max} &\geq& \! x_{12}^* + |\{x_{12}(b_1)\} \cap \{x_{12}(b_2)\}| - 1 \nonumber\\
 &>& \, t_{12} + |\{x_{12}(b_1)\} \cap \{x_{12}(b_2)\}|-1  \label{16}\\
 &\geq & \! \min(t_{12}+t_2-t^1_{12}, t_{12}+t_1-t^2_{12}) \label{17}\\
 &\geq & \! \min(t_{12}^1+t_{12}^2+t_2-t^1_{12}, t_{12}^1+t_{12}^2+t_1-t^2_{12}) \label{thinh18}\\
 &= & \! \min(t_2+t^2_{12}, t_1+t^1_{12}) \nonumber\\
 &= & \! \min(c_2-1,c_1-1), \nonumber
\end{eqnarray}
\end{small}
where (\ref{16}), (\ref{17}) and (\ref{thinh18}) are due to (\ref{11}), (\ref{12}) and (\ref{19}), respectively. Therefore $x_{12}^{\max}$ is strictly greater than $\min(c_2-1,c_1-1)$. But this contradicts with the way we constructed the set $\{x_{12}(b_1)\} \cap \{x_{12}(b_2)\}$ whose maximum element cannot exceed $\min(c_1-1,c_2-1)$ due to $\mod c_1$ and $\mod c_2$ operation in the encoding procedure. Therefore,$x^*_{12}$ must satisfy condition (b).

%
Next, due to $x_{12}^* \in \{x_{12}(b_1)\} \cap \{x_{12}(b_2)\} $ and from (\ref{13}), (\ref{14}), we have:
$$b_i-x_{12}^* \in \{0,1,\dots,t_i\} {\mod (c_i)}$$

Therefore, from (\ref{15}):
\begin{equation}
\label{eq:x_i}
x_i^* = b_i - x_{12}^* \in \{0,1,\dots,t_i\} {\mod (c_i)}.
\end{equation}

This establishes the verification for (c).

The correctness of condition (d) can be easily seen by noting that $b_i = \hat{b}_i$ by combining Eqs. (\ref{eq:received_optimal}),  (\ref{eq:decoding}), and  (\ref{eq:x_i}).
\end{proof}

\begin{exmp}
\label{example_1}
To illustrate Proposition \ref{prop:JRC}, we will show an example of a topology consisting of three transmitters and two receivers shown in Fig. \ref{fig:n_3_m_2}(a).  The number of exclusive transmitters for $R_1$ and $R_2$ are $t_{1}=0$ and $t_{2}=2$ while the number of shared transmitters $t_{12}=1$.
Choose $t^1_{12} = 1$ and $t^2_{12}=0$, then this pair is valid since:

$$
\begin{cases}
t^1_{12},t^2_{12}\geq 0, \\
t^1_{12}+t^2_{12}\leq t_{12}=1, \\
t_{1}\geq t^2_{12},\\
t_{2}\geq t^1_{12}.
\end{cases}
$$

Then, from Proposition \ref{prop:JRC}, the  achievable rate of $R_1$ is $\log({t_1}+t^1_{12}+1)=\log {(c_1)}=\log{(2)}$, and for $R_2$ is  $\log(t_2+t^2_{12}+1)=\log{(c_2)}=\log{(3)}$.
Therefore, $R_1$, $R_2$ can achieve arbitrary pattern $(b_1,b_2)$ with $b_1 \in \{0,1\}$ and $b_2 \in \{0,1,2\}$, respectively.

To illustrate the encoding and decoding procedures, suppose that $(b_1,b_2)=(1,2)$ is desired pattern in $R_1$, $R_2$. Then encoding and decoding procedure will be presented as below to find $(x_1^*, x_2^*, x_{12}^*) = \mathcal{C}(b_1, b_2)$.

\textbf{Encoding}: the encoding procedure will construct two sets:
\begin{eqnarray*}
\{x_{12}(b_1)\} &=& \{1-i_1 \mod (2), i_1 = 0\} =\{ 1\}. \\
\{x_{12}(b_2)\} &=& \{2-i_2 \mod (3), i_2 = 0, 1, 2\}=\{2, 1, 0 \}.
\end{eqnarray*}

Then, $\{x_{12}(b_1)\} \cap \{x_{12}(b_2)\} = \{1\}$. Choose $x_{12}^*=1$. Next, construct $x_1$ and $x_2$ as:
$$x_1^*=b_1-x_{12}^*=1-1=0 \mod (2).$$
$$x_2^*=b_2-x_{12}^*=2-1=1 \mod (3).$$

Hence, $(x_1^*,x_2^*,x_{12}^*)=(0, 1, 1)$.

\textbf{Decoding}: the decoding procedure will decode by summing up all received signals at each receiver, ie,:
$$\hat{b}_1=x_1^*+x_{12}^*=0 + 1 = 1 \mod(2) = b_1.$$
$$\hat{b}_2=x_2+x_{12}=1 + 1 = 2 \mod(3) = b_2.$$
\end{exmp}

Similar to ERC method, the JRC method can be extended to arbitrary number of receivers. Next, we will present the extended results for $n$ receivers with pairwise sharing transmitters.

\begin{proposition}{(Achievable rates for $n$-receiver pairwise sharing transmitter topology)}
\label{prop:JRC extend}
Given a topology consisting of $n$ receivers $R_1,R_2,\dots,R_n$, if each receiver $R_i$ has $t_i$ exclusive transmitters and $t_{ip}$ sharing transmitters with other receiver $R_p$. Then the receiver $R_i$ can achieve the rate:
 $$\log(c_i^n)=\log{(t_i + \sum_{p \neq i;p=1}^{p=n}{{t}^i_{ip}} + 1)}.$$
bits per time slot in which $i$ is the notation for the receiver $R_i$ and $n$ is the number of receiver in network if with $\forall p \in \{1, \dots ,n\}$ and  $p\neq i $:
\begin{equation}
{t^i_{ip}} \leq {t_{p}} \label{21}.
\end{equation}

Note: In the case $t_{ip}=0$, i.e.,  $R_i$ and $R_p$ do not share any transmitter, then in the inequality, $t_{p}$ will be replaced by ``0" or the number of sharing transmitters assigned to $R_i$ is ${t^i_{ip}}=0$.
\end{proposition}

We also note that Proposition \ref{prop:JRC extend} is only applicable to topologies with pair-wise sharing transmitters only, i.e., any transmitter can cover at most two receivers.  Furthermore, the rate region for all the receivers are specified by the tunable values $t^i_{ip}$ such that the conditions in Proposition \ref{prop:JRC extend} are satisfied for all $i$ and $p$.  The larger $t^i_{ip}$ will allow the receiver $R_i$ to obtain a larger rate at the expense of a reduced rate for $R_p$.

From two receivers $R_i$ and $R_p$ perspective, Proposition \ref{prop:JRC extend} states that receiver $R_i$ can be allocated ${t^i_{ip}}$ transmitters from $t_{ip}$ sharing transmitters between $R_i$ and $R_p$  if:
 $${t^i_{ij}} \leq t_j.$$
Therefore, by applying Proposition \ref{prop:JRC extend} for all receivers $R_1,R_2,\dots,R_n$, we can solve and distribute suitable rates for all receivers in a given topology.  The proof of Proposition \ref{prop:JRC extend} is shown below.
\begin{proof}

The proof is based on induction.  For the basis case of two receiver topology ($n = 2$) is true from Proposition \ref{prop:JRC}.  Now,  suppose that Proposition \ref{prop:JRC extend} holds for $n-1$ receiver topology, we will show that Proposition \ref{prop:JRC extend} will also hold for $n$ receiver topology where one more receiver $R_n$ is added to the topology.  Fig. \ref{fig: n-1set} illustrates the inductive method.

First, using Proposition \ref{prop:JRC extend} with $n-1$ receivers topology, receiver $R_i$ with $i \in \{1,\dots,n-1\}$ can achieve the rate:
$$\log(c_i^{n-1})=\log{(t_i + \sum_{p\neq i;p=1}^{p=n-1}{{t}^i_{ip}} + 1)}.$$

It means that receiver $R_i$ is able to distinguish all value in set $\{0, 1, \dots c_i^{n-1}-1\}$.

After adding receiver $R_n$ with $t_n$ exclusive transmitters into network and $t_{in}$ ($i = 1, 2, \dots, n-1$) sharing transmitters, for  Proposition \ref{prop:JRC extend} to hold, we need to verify two following conditions:

\textbf{Condition (a)}: all previous receivers $R_i$ with $i \in \{1,\dots,n-1\}$ can obtain additional ${t}_{in}^i$ states, and therefore achieve the new rates:
\begin{eqnarray*}
\log(c_i^n)&=&\log{(t_i+\sum_{p\neq i;p=1}^{p=n-1}{{t}^i_{ip}}+1 +{t}_{in}^i)}\\
&=&\log{(c_i^{n-1} + {t}_{in}^i)}.
\end{eqnarray*}
Hence,
\begin{eqnarray}c_i^n = c_i^{n-1} + {t}_{in}^i. \label{20}
\end{eqnarray}

To do so, we need to verify that receiver $R_i$ is able to distinguish all values in the set $\{0, 1, \dots c_i^n-1\}$.

\textbf{Condition (b)}: the new receiver $R_n$ also satisfies Proposition \ref{prop:JRC extend}, i.e., $R_n$ is able to achieve the rate:

$$\log(c_n^n)=\log{(t_n + \sum_{p=1}^{p=n-1}{{t}^n_{np}} + 1)}.$$

We  first verify condition (a). Suppose that we need to transmit a signal $b_i$ to the receiver $R_i$, with:
$$b_i \in \{0, 1, \dots c_i^n-1\}.$$

Let us divide $b_i$ into two subsets:

$\bullet$ If $ 0 \leq b_i \leq {c_i^{n-1}-1}$: We will transmit $b_i$ in the $n-1$ previous receiver topology (using the the previous transmitters) and sends ``0" using $t_{in}$ sharing transmitters with new receiver $R_n$. Clearly, receiver $R_i$ will receive correct pattern since by assumption,Proposition \ref{prop:JRC extend} holds true for $n-1$ receiver topology.

$\bullet$ If  $c_i^{n-1}-1 < b_i \leq c_i^n-1$: We will transmit signal $c_i^{n-1}-1$ in the $n-1$ previous receiver topology and send the signal:
$$x_{in}=b_i-(c_i^{n-1}-1) \mod (c_i^n)$$
using the new $t_{in}$ sharing transmitters. Clearly,
\begin{eqnarray}
 x_{in}&=&b_i- (c_i^{n-1}-1) \nonumber \\
  	   &\leq & (c_i^n -1) - (c_i^{n-1}-1)\nonumber\\
       &=& t_{in}^i \label{18} \\
       &\leq & t_{in} \label{22}.
\end{eqnarray}

With (\ref{18}) is due to (\ref{20}), and:
\begin{eqnarray}
 x_{in}&=&b_i-(c_i^{n-1}-1)\nonumber\\
       &\geq& (c_i^{n-1}-1) - (c_i^{n-1}-1)\nonumber\\
       &=& 0 \label{23}.
\end{eqnarray}

From (\ref{22}) and (\ref{23}): $0\leq x_{in} \leq t_{in}^i \leq t_{in}$, then $t^i_{in}$ sharing transmitters can always transmit the signal $x_{in}$. Consequently, the received signal at $R_i$ is $y_i = c_i^{n-1} - 1 + x_{in}$ (Note that $c_i^{n-1} - 1$ comes from the transmitters in previous topology).  Using the same decoding method as in Eq. (\ref{eq:decoding}), we have:
\begin{eqnarray}
\hat{b}_i &=& y_i \mod (c^n_i) \\
&=& c^{n-1}_{i} - 1 + x_{in} \mod (c^n_i) \\
&=& c^{n-1}_i - 1 +  b_i - (c^{n-1}_i-1) \mod(c^n_i) \\
&=& b_i.
\end{eqnarray}

Therefore,  the previous receiver $R_i$ can distinguish all values of $b_i \in \{0, 1, \dots c_i^n-1\}$ and achieve the rate $\log(c_i^n)$ with:
\begin{equation}
\label{ri}
\log(c_i^n)=\log{(t_i + \sum_{p\neq i;p=1}^{p=n}{{t}^i_{ip}} + 1)}.
\end{equation}

Next, we verify condition (b) that the new receiver $R_n$ also satisfies Proposition \ref{prop:JRC extend}, i.e., receiver $R_n$ is able to achieve the rate:
$$\log(c_n^n)=\log{(t_n + \sum_{p=1}^{p=n-1}{{t}^n_{np}} + 1)}.$$

Indeed, for a fixed pattern $b_i$ with $i=1,\dots,n-1$ in the $n-1$ old receivers, we will prove that receiver $R_n$ can discern $c_n^n$ states:
$$b_n \in \{0,1,\dots,c_n^n-1\}$$

Let observe the receiver $R_i$, with fixed pattern $b_i$ as in Fig. \ref{fig: n-1set}. We note that of the $t_{in}$ sharing transmitters between $R_i$ and $R_n$, $t_{in}^i$ transmitters are allocated to $R_i$ and $t_{in}^n$  remaining transmitters will be distributed to $R_n$. Now, we can maintain the pattern $b_i$ by transmitting the pattern $(b_i-\delta_i) \mod c^{n-1}_{i}$ using the transmission method as described in condition (a), then transmit pattern $\delta_i$ in $t_{in}^n$ remaining transmitters, where

$$0 \leq \delta_i \leq t_{in}^n, $$

since the number of levels in $\delta_i$  cannot exceed the number of transmitters.

Now, from  condition (\ref{21}) from Proposition \ref{prop:JRC extend} to the pairwise sharing transmitter between $R_n$ and $R_i$, we have:
$$t_{in}^n \leq t_i.$$

Therefore,
$$0 \leq \delta_i \leq t_{in}^n \leq t_i.$$

The inequality above together with the $t_n$ exclusive transmitters show that $R_n$ is able to achieve $(t_{in}^n+1)$ distinguishable states in pairwise sharing transmitter between $R_i$ and $R_n$ when:
$$\delta_i \in \{0,1,\dots,t_{in}^n\}.$$

Thus, for the all shared transmitters between $R_1, R_2,\dots, R_{n-1}$ with $R_n$ and $t_n$ exclusive transmitters of $R_n$, the number of distinguishable levels at $R_n$ is:
$$c_n^n = t_n + \sum_{p=1}^{p=n-1}{t_{np}^n} + 1.$$
Then, the achievable rate can be achieved in $R_n$ is:
\begin{equation}
\label{rn}
\log(c_n^n)=\log{(t_n + \sum_{p=1}^{p=n-1}{{t}^n_{np}} + 1)}.
\end{equation}

\end{proof}


\begin{figure}
\hspace*{-0cm}
  \centering
  \includegraphics[width=2.8in]{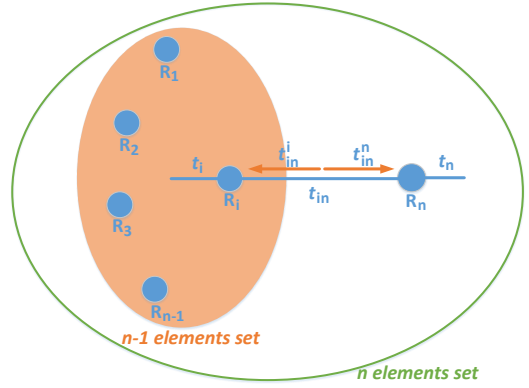}
  \caption{Inductive method from $n-1$ element set to $n$-element set }\label{fig: n-1set}
  \vspace{-.1in}
\end{figure}

In practice, there are many deployments that are not pairwise sharing topologies.  We have a simple following result regarding the multi-user capacities:

\begin{proposition}
\label{prop:limit 2}
Given an arbitrary topology with $k$ transmitters and $n$ receivers $R_1,R_2,..., R_n$.   If each receiver $R_i$ has an achievable rate $\log{(c_i^n)}$ bits per time slot, then
$$\sum_{i=1}^{i=n}{\log {c_i^n}} \leq k.$$
\end{proposition}
\begin{proof}

Since the maximum bit rate can generate by all $k$ transmitters is $k$ bits per second.  This total rate must be shared among all the receivers.  Thus, the proof follows.

\end{proof}

{\bf General Topology.} Proposition \ref{prop:limit 2} is less useful since the described achievable rate region does not exploit the topological information.
In what follows, we describe a very simple algorithm for converting many non-pairwise sharing topologies into a pair-wise sharing topology whose achievable rate region can be characterized.   In particular, a general topology consisting of $k$ transmitters and $n$ receivers can be characterized by collection of sets of different types of transmitters: exclusive transmitters, pairwise sharing transmitters,3-sharing transmitters, ..., $n$-sharing transmitters.

Initially, we construct the pairwise sharing topology that is characterized by all the exclusive and pairwise sharing transmitters from the set of all the transmitters.  If the condition of Proposition \ref{prop:JRC extend} is satisfied, then the achievable region for this pairwise sharing topology can be characterized.  Now, the achievable region for a new topology that includes the existing pair-wise sharing topology and one additional $n$-sharing transmitter ($n > 2$) can be computed as follows.  Suppose this new transmitter is shared among $R_1, R_2, \dots, R_m$ receivers.  Then we can assign this new transmitter to a pair of receivers in $(R_1, R_2, \dots, R_m)$.  Suppose $R_i$ and $R_j$ were chosen, then the number of shared transmitters for this pair $t_{R_iR_j}$ is increased by one.  Effectively, we have a new pairwise sharing topology.

However since a transmission by new shared transmitter will affect the receivers $R_1, R_2, \dots, R_m$, we need to modify the encoding procedure slightly.  First, if the new transmitter $t_{R_iR_j}$ transmits bit "0", the encoding procedure for the bit pattern $b_i$ intended for $R_i$ is the same as one used for the pair-wise sharing topology without the new shared transmitter. This is because the bit "0" does not interfere with other signals.   If $t_{R_iR_j}$ transmits bit "1", then to transmit the original bit pattern $b_l$ intended for receiver $R_l$, $ l \neq i, j$, we encode $b_l - 1$ instead using the same encoding (transmission) procedure for the pair-wise sharing topology without $t_{R_i,R_j}$.  Similar to the proof for Proposition \ref{prop:JRC extend}, specifically condition (b), it is to see that all the receiver $R_l$, $l \neq i,j$ will be able to recover original bit pattern $b_l$.  Specifically, either receivers $R_i$ or $R_j$ will increase its capacity to $\log(c_i + 1)$ or  $\log(c_i + 1)$, depending on whether $t_{R_i,R_j}$ is assigned to $R_i$ or $R_j$, while other receivers will have the same capacities as before.

{\bf Maximum Sum Rate.} Generally, the procedure of adding a new shared transmitters is repeated and the corresponding achievable regions can be characterized if the conditions in Proposition \ref{prop:JRC extend} are satisfied.  We also note that there are exponential large number of ways that the shared transmitters can be assigned to receivers, but the number of valid assignments based on Proposition \ref{prop:JRC extend}, are generally a lot smaller. On the other hand, to maximize the sum rate of all the receivers, we have a greedy algorithm for determining which receiver should get a new shared transmitter during the allocation.  Specifically, we will allocate the shared transmitter to the receiver with smallest rate at every step for the following reason.

If we allocate a shared transmitter $t_{R_i,R_j}$ to $R_i$ which currently has an achievable rate $\log(c_i)$, then the capacity gain for $R_i$ is:
$$\log(c_i+1)-\log(c_i)=\log(1+1/c_i).$$
Similarly if we allocate a shared transmitter $t_{R_i,R_j}$ to $R_j$, then the capacity gain for $R_j$ is:
$$\log(c_j+1)-\log(c_j)=\log(1+1/c_j).$$
Clearly, $\log(1+1/c_i) \ge \log(1+1/c_j)$ if $c_i \le c_j$. So, we should allocate the shared transmitter to the receiver with the smallest capacity currently if we want largest gain in one step (greedy) in capacity.

\begin{exmp}
\label{example_2}
 We use this example to illustrate the procedure for converting a non-pair-wise sharing topology to pair-wise sharing topology and obtain a point in the achievable rate region. Fig. \ref{fig: convert}(a) represents a non-pairwise sharing topology  with $t_1=1,t_2=1,t_3=2,t_{12}=t_{23}=t_{31}=2$, and  $t_{123}=1$.

Suppose we allocate $t_{123}$ to the pair $(R_1,R_3)$. Applying the aforementioned conversion procedure, we obtain the resulted pair-wise topology shown in Fig. \ref{fig: convert}(b) with:
$$t_{13}'=t_{13}+ 1 =3.$$
\begin{figure}
\hspace*{-1 cm}
  \centering
  \includegraphics[width=3.5in]{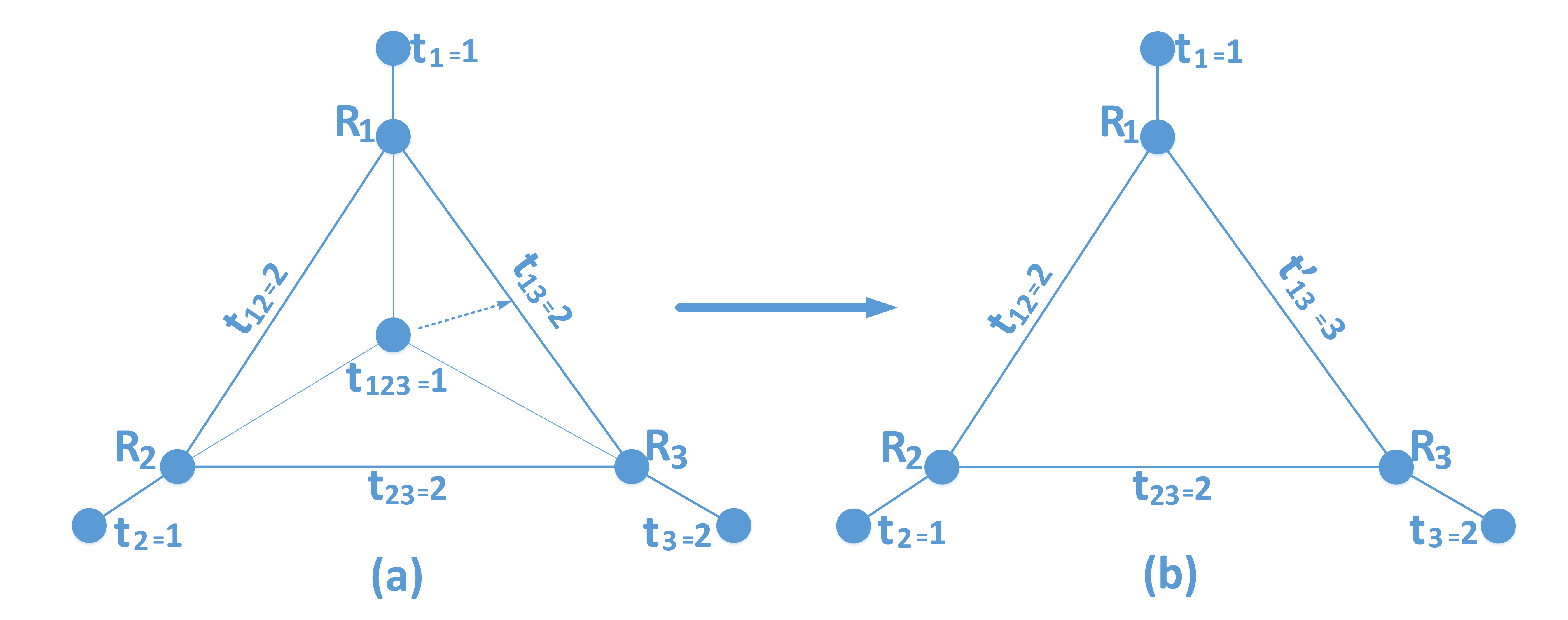}
  \caption{Convert high level connection to two level connection }\label{fig: convert}
  \vspace{0.1 in}
\end{figure}

Now we have a choice of selecting value for ${t_{13}^1}'$ and ${t_{13}^1}'$ .  However,  based on Proposition \ref{prop:JRC extend}, the following constraints must hold:
$$
\begin{cases}
t_{12}^1+t_{12}^2\leq t_{12}=2, \\
{t_{13}^1}'+{t_{13}^3}'\leq t_{13}'=3, \\
t_{23}^2+t_{23}^3\leq t_{23}=2, \\
0 \geq t_{12}^1 \leq t_2 =1,\\
0 \geq t_{12}^2 \leq t_1 =1,\\
0 \geq {t_{13}^1}' \leq t_3 =2,\\
0 \geq {t_{13}^3}' \leq t_1 =1,\\
0 \geq t_{23}^2 \leq t_3 =2,\\
0 \geq t_{23}^3 \leq t_2 =1.\\
\end{cases}
$$

All the pairs of $(t_{12}^1$, ${t_{13}^1}'$, $t_{12}^2$, $t_{23}^2$, ${t_{13}^3}'$, $t_{23}^3)$ that can satisfy the above constrains are valid to distributed to receivers $(R_1,R_2,R_3)$. For example the pairs $t_{12}^1=1$, $t_{12}^2=1$, ${t_{13}^1}'=2$, ${t_{13}^3}'=1$, $t_{23}^2=1$, $t_{23}^3=1$ are valid. Hence, $R_1,R_2$ and $R_3$ can achieve the rate $\log{(5)},\log{(4)}$ and $\log{(5)}$ bit per time slot, respectively.

 As an example to illustrate the encoding process when using a non-pairwise sharing transmitter.  Suppose that we want to transmit the pattern $(b_1=2,b_2=3,b_3=5)$ to $(R_1,R_2,R_3)$, respectively. Based on the conversion procedure discussion, there are two cases to consider: $x_{123}=0$ and $x_{123}$ = 0.

$\bullet$ Suppose $x_{123}=0$, then based on the encoding in tge conversion procedure, the pattern $(b_1=2,b_2=3,b_3=5)$ is transmitted normally. Using Proposition \ref{prop:JRC extend}, we construct $n=3$ sets according the encoding procedure:
$$
\begin{cases}
x_{12}+x_{13} \in \{b_1 - i_1 , i_1=0, 1\} = \{2, 1\} \mod(5),\\
x_{12}+x_{23} \in \{b_2 - i_2 , i_1=0, 1\} = \{3, 2\} \mod(4),\\
x_{13}+x_{23} \in \{b_3 - i_3 , i_1=0, 1, 2\} = \{5, 4, 3\} \mod(5),\\
0 \leq x_{12} \leq 2, \\
0 \leq x_{13} \leq 2, \\
0 \leq x_{23} \leq 2.
\end{cases}
$$
Next, a set of feasible solution to the above inequalities is:
$$
\begin{cases}
x_{12}= 0, \\
x_{13}= 2, \\
x_{23}= 2, \\
i_1=x_1= 0, \\
i_2=x_2= 1, \\
i_3=x_3= 1.
\end{cases}
$$

Now, we note that the decoding procedure sums up all the signal at the receiver:
$$
\begin{cases}
b_1=x_1+x_{12}+x_{13}+x_{123}= 2, \\
b_2=x_2+x_{12}+x_{23}+x_{123}= 3, \\
b_3=x_3+x_{13}+x_{23}+x_{123}= 5.
\end{cases}
$$
As seen, they are all correct.

$\bullet$ Suppose $x_{123}=1$. Then based on the encoding in the conversion procedure, the pattern $(b_1=1,b_2=2,b_3=4)$ is transmitted. Using Proposition \ref{prop:JRC extend}, we construct $n=3$ sets based on the encoding procedure:
$$
\begin{cases}
x_{12}+x_{13} \in \{b_1 - i_1 , i_1=0, 1\} = \{1, 0\} \mod(5),\\
x_{12}+x_{23} \in \{b_2 - i_2 , i_1=0, 1\} = \{2, 1\} \mod(4),\\
x_{13}+x_{23} \in \{b_3 - i_3 , i_1=0, 1, 2\} = \{4, 3, 2\} \mod(5),\\
0 \leq x_{12} \leq 2, \\
0 \leq x_{13} \leq 2, \\
0 \leq x_{23} \leq 2.
\end{cases}
$$
Next, a set of feasible solution to the inequality above is:
$$
\begin{cases}
x_{12}= 0, \\
x_{13}= 1, \\
x_{23}= 1, \\
i_1=x_1= 0, \\
i_2=x_2= 1, \\
i_3=x_3= 2.
\end{cases}
$$

Now, the decoding procedure sums up all the signal go to receiver:
$$
\begin{cases}
b_1=x_1+x_{12}+x_{13}+x_{123}= 2, \\
b_2=x_2+x_{12}+x_{23}+x_{123}= 3, \\
b_3=x_3+x_{13}+x_{23}+x_{123}= 5,
\end{cases}
$$
to correctly reconstruct the transmitted patterns.
\end{exmp}
\subsection{Achievable rate Region for ideal channels}
In this section, we will use LAC cooperative transmission techniques SRC, ERC, and JRC to characterize the achievable rate regions for common topologies.

\subsubsection{Achievable Rate Region for Two-Transmitter Topologies}
\label{sec:two-transmitter}
For the two-transmitter topologies with the number of receivers  being smaller than the number of transmitters, there are only two canonical topologies shown in Fig. \ref{fig:n_2_m_2_lac}.  Other topologies where receivers are not in an overlapped region are trivial.

\begin{figure}[h]
\begin{center}
$\begin{array}{cc}
\includegraphics[scale=0.23]{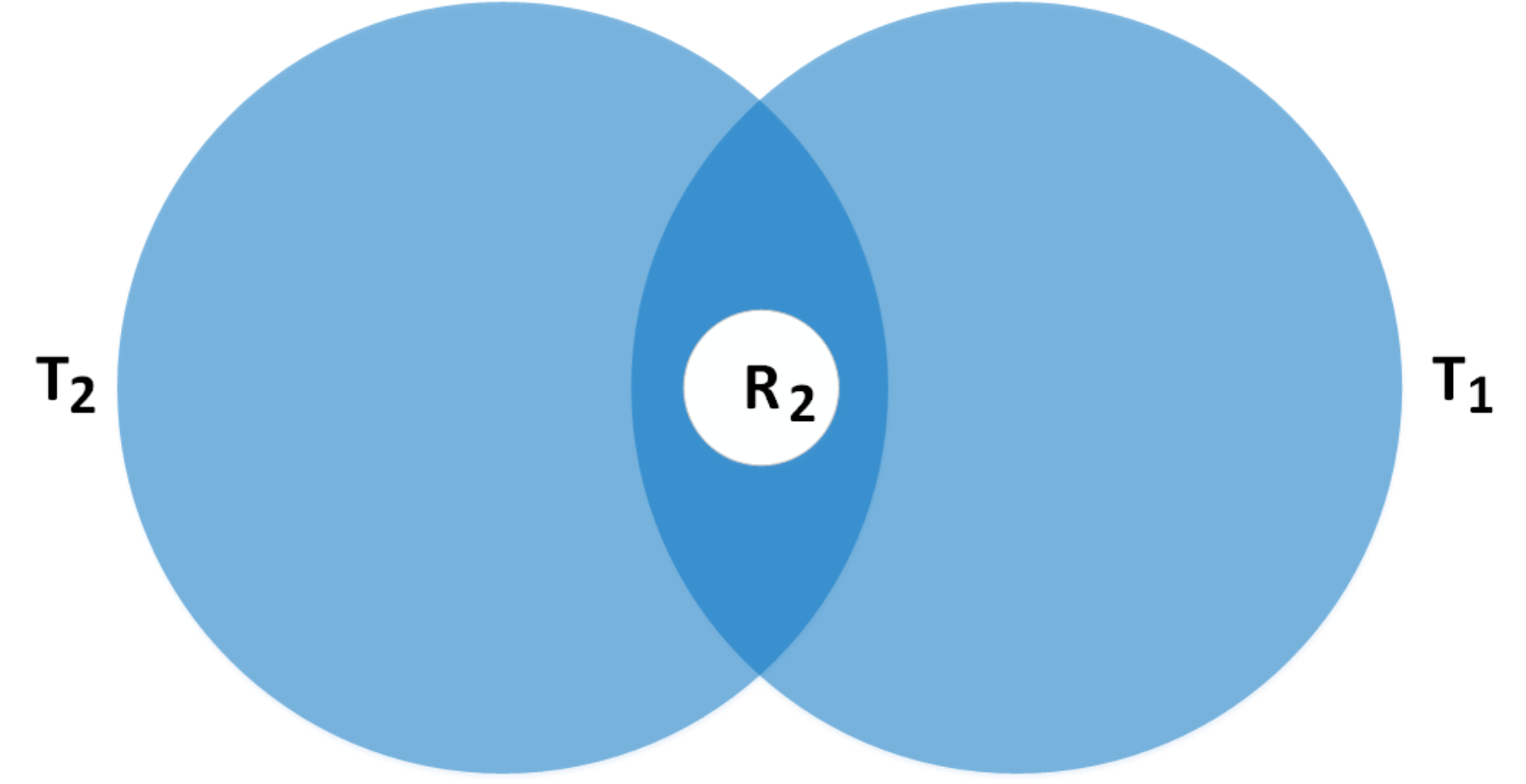} &
\includegraphics[scale=0.23]{LAC21-eps-converted-to}\\
(a) & (b)
\end{array}$
\end{center}
\caption{Topologies for (a) two transmitters and one receiver; (b) two transmitters and two receivers.}
\label{fig:n_2_m_2_lac}
\end{figure}

As discussed in Section \ref{sec:problem}, using time-sharing scheme between $R_1$ and $R_2$, the achievable rate region is depicted as the blue triangle in Fig. \ref{fig:n_2_m_2 rate_region} with its boundary being a linear interpolation of two achievable rate tuples (0,1) and (1,0).
Now, using SRC (Proposition \ref{prop:SRC}) for $R_2$ and $R_1$, rate tuples $(0, \log{3})$ and $(1,0)$ are achievable. Thus, SRC helps  enlarge the achievable region by additional green area.  The achievable region can be further enlarged by an additional yellow area by using SRC for $R_2$ to obtain the rate tuple $(0,\log{3})$  and ERC (Proposition \ref{prop:ERC}) for both $R_1$ and $R_2$ to obtain the rate tuple (1,1).  Consequently, the achievable rate region is obtained by interpolation between the two rate tuples $(0, \log{3})$ and (1,1).

\begin{figure}
  \centering
  \includegraphics[width=2.4in]{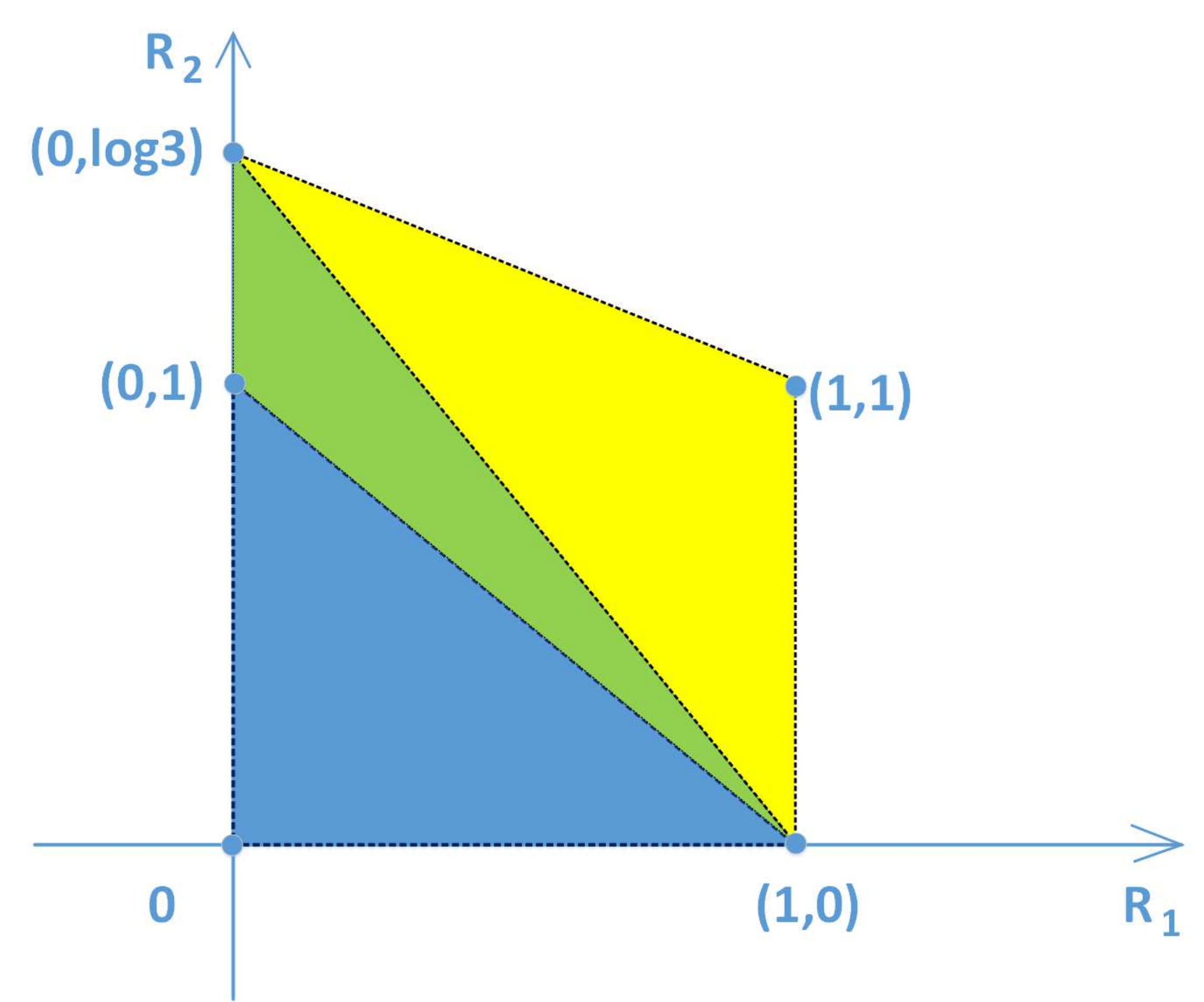}
  \caption{Achievable rate region using SRC scheme for $R_2$ and ERC scheme for both $R_1$ and $R_2$.}\label{fig:n_2_m_2 rate_region}
  \vspace{-.2in}
\end{figure}

\subsubsection{Achievable Rate Region for Three-Transmitter Topologies}
\label{sec:three-transmitter}
Similar to the two-transmitter topologies, the achievable rate region of the three-transmitter topologies is constructed by finding the feasible tuples that can be achieved using SRC and  ERC, and additionally JRC.
Specifically, for three-transmitter topologies, the canonical topologies with the number receivers: 1, 2, and 3,  are as shown in Fig. \ref{fig:n_3_m_2_lac}.  First, using SRC (Proposition \ref{prop:SRC}) for $R_3$, $R_2$ and $R_1$, rate tuples $(0,2,0)$, $(\log{3},0,0)$ and $(0,0,1)$ are achievable as the red triangle shown in Fig.  \ref{fig:Diagram_6}.  Note that the $x, y$, and $z$ coordinates denote the rate for $R_2$, $R_3$, and $R_1$, respectively.
Next, using ERC (Proposition \ref{prop:ERC}), the feasible tuple (1,1,1) can be obtained.  Thus, the achievable region is enlarged  as shown by the  green pyramid with four vertices $(0,2,0)$, $(\log{3},0,0)$ and $(0,0,1)$, and  $(1,1,1)$.

\begin{figure}[h]
\begin{center}
$\begin{array}{cc}
\includegraphics[scale=0.21]{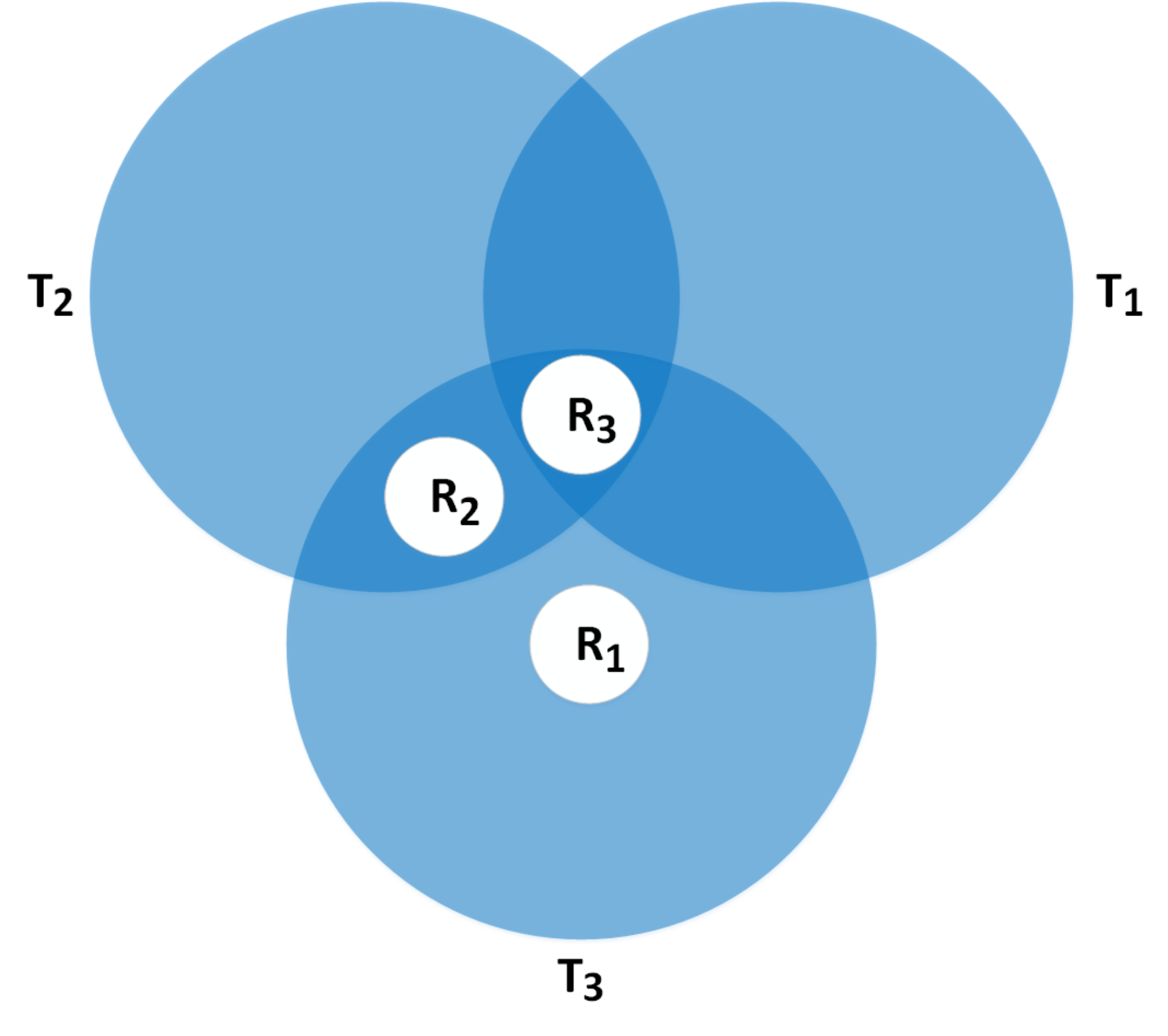} &
\includegraphics[scale=0.21]{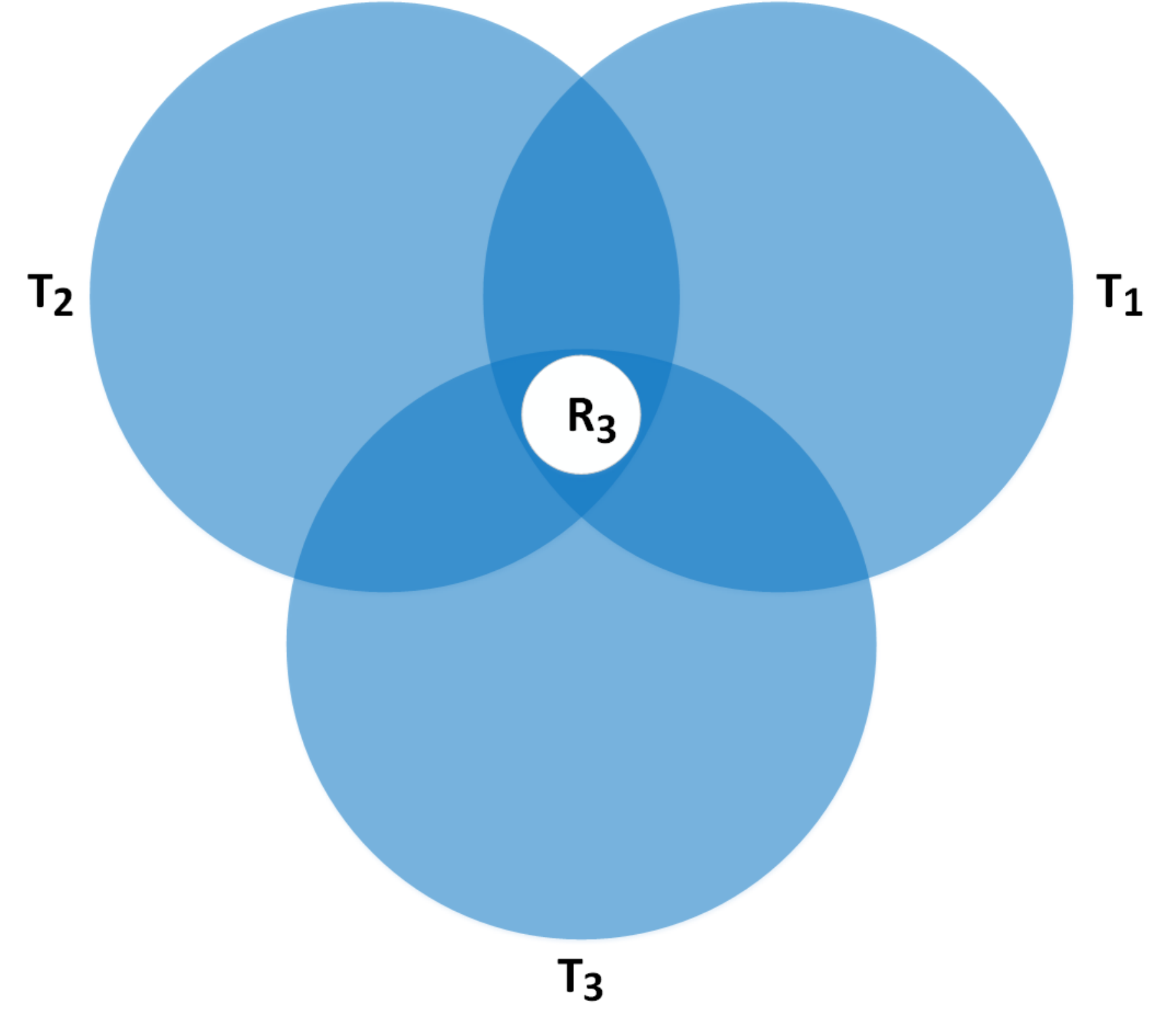}\\
(a) & (b) \\
\includegraphics[scale=0.2]{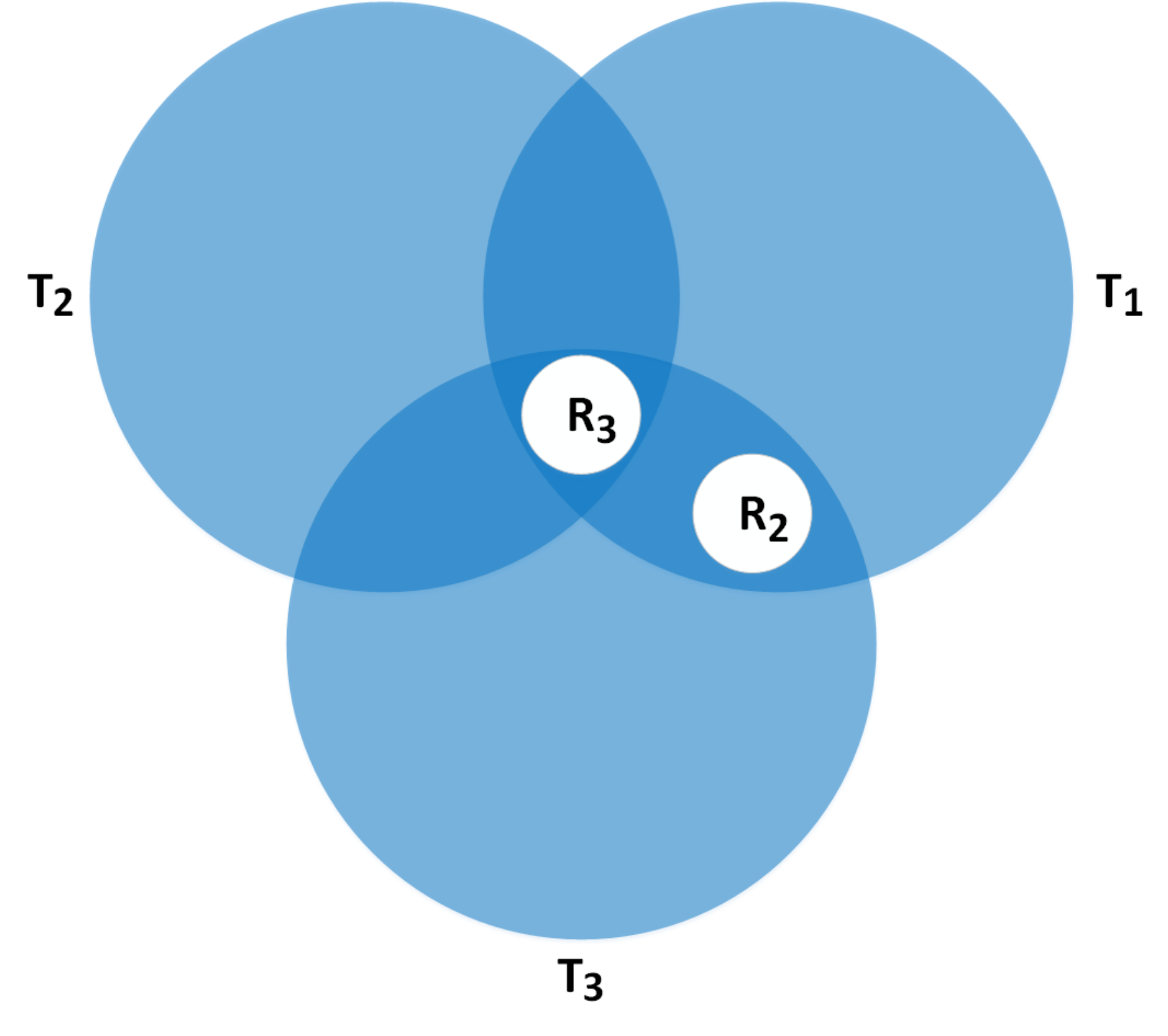} &
\includegraphics[scale=0.154]{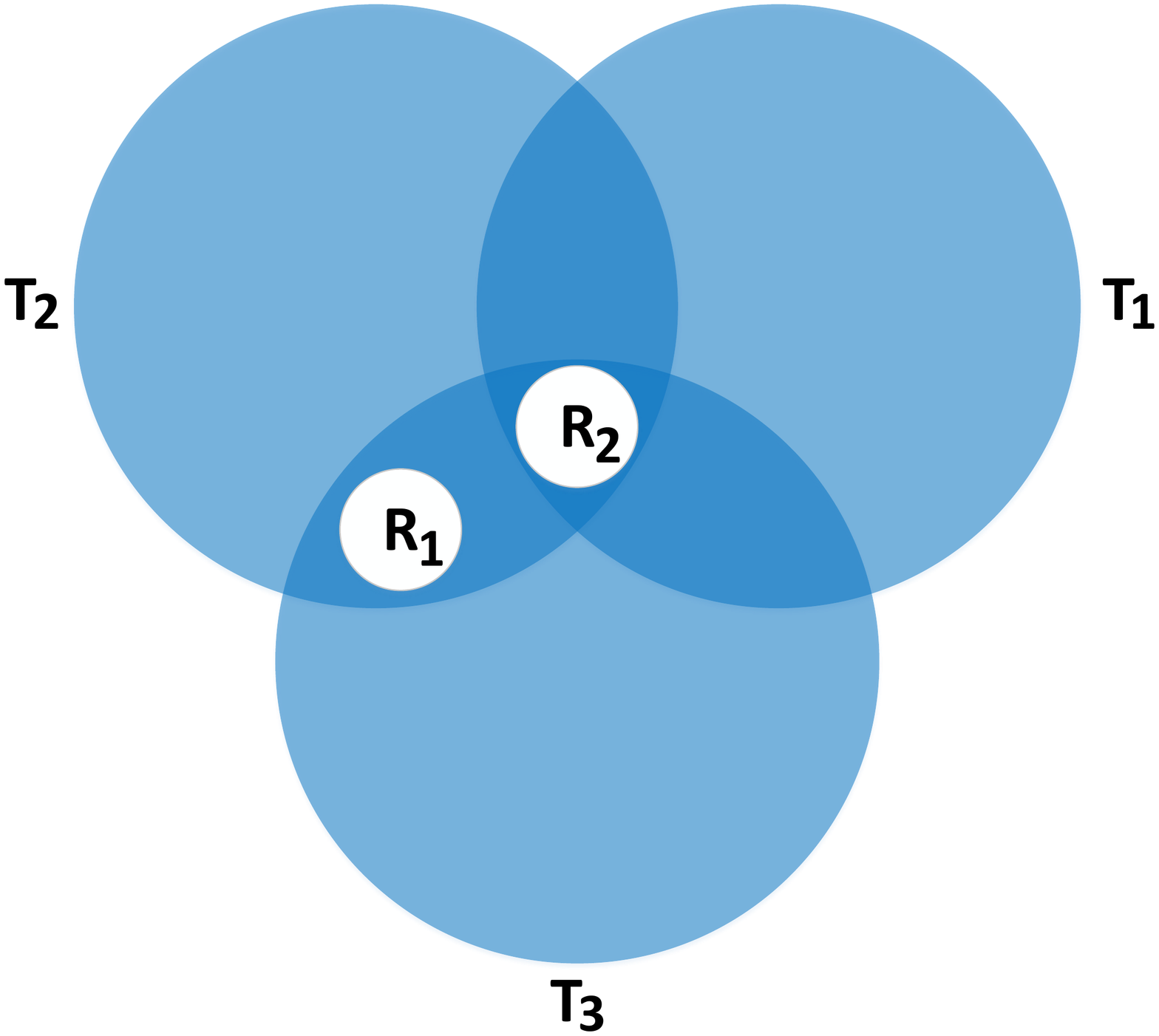} \\
(c) & (d)
\end{array}$
\vspace{-.15in}
\end{center}
\caption{Topologies for (a) three transmitters and three receivers; (b) three transmitters and one receiver;  (c) and (d) three transmitters and two receivers.}
\label{fig:n_3_m_2_lac}
  \vspace{-.1in}
\end{figure}

Next, by applying JRC (Proposition \ref{prop:JRC})  for topologies in Fig. \ref{fig:n_3_m_2_lac} (c) and (d),  the two tuples $(\log{3}, 1, 0)$  and $(0, \log{3}, 1)$  can be obtained, respectively.
Specifically, for the tuple $(0, \log{3}, 1)$,  the number of exclusive transmitters for $R_1$ and $R_3$ are $t_1=0$ and $t_3=2$ while the number of shared transmitters $t_{13}=1$.  Using Proposition \ref{prop:JRC} with $t_{13}^1 = 1$ and $t_{13}^3=0$, the  achievable rate of $R_1$ is $\log(t_1+t_{13}^1+1)=1$, and for $R_3$,  $\log(t_3+t_{13}^3+1)=\log(3)$.  Using the same technique for $R_3$ and $R_2$ shown in  Fig. \ref{fig:n_3_m_2_lac} (c),  the feasible tuple $(\log{3}, 1, 0)$ can be obtained.


Finally, Fig. \ref{fig:Diagram_6} shows the overall achievable rate region as a convex hull of the feasible tuples:   $(0, \log{3}, 1)$,   $(\log{3}, 1, 0)$, $(1,1,1)$, $(0,0,1)$, $(\log{3}, 0, 0)$, $(0,2,0)$.


\begin{figure}
\hspace*{-0.1cm}
  \centering
  \includegraphics[width=2.3in]{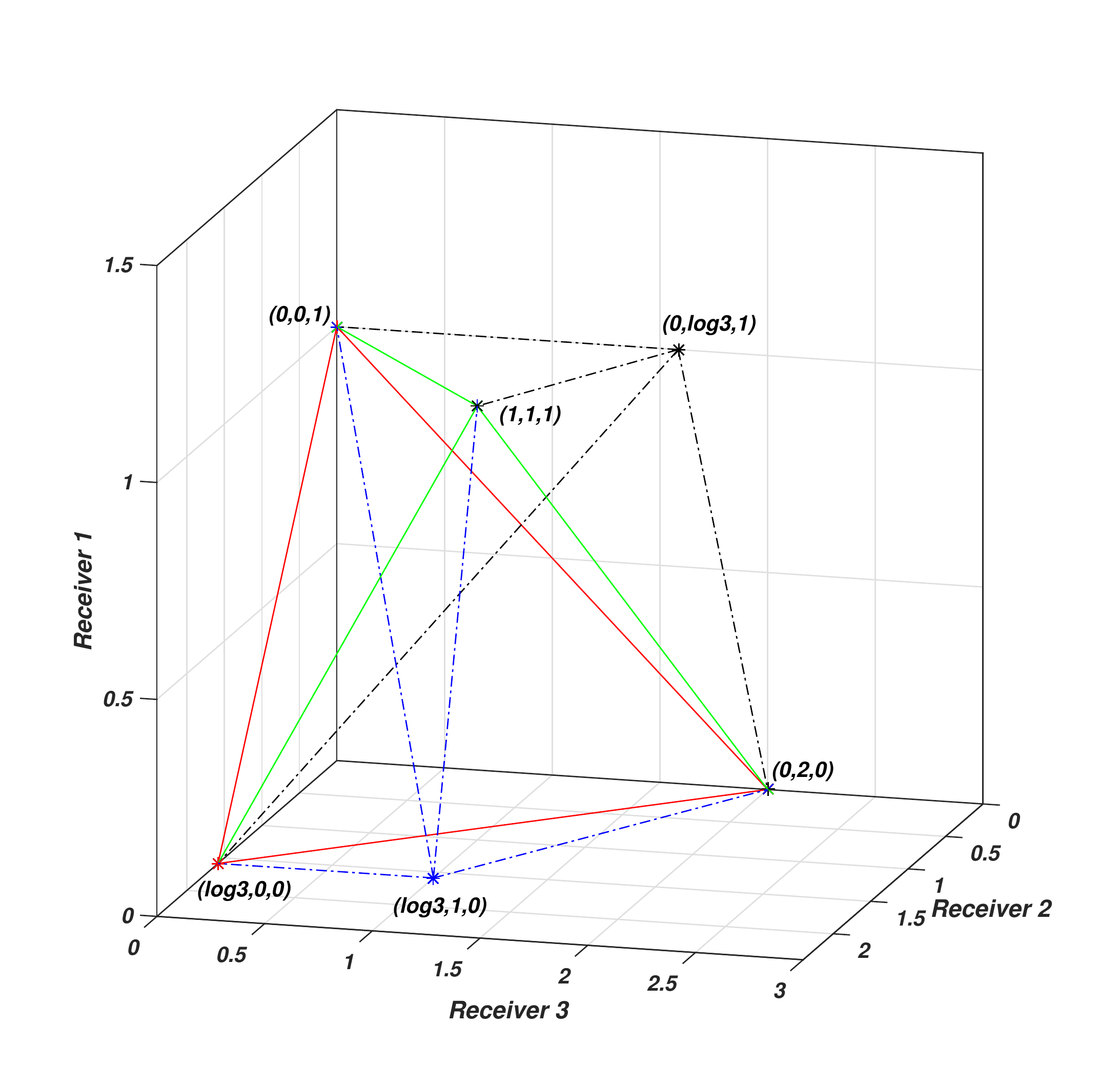}
  \caption{Achievable rate region for three transmitters topology}
  \label{fig:Diagram_6}
  \vspace{-.25in}
\end{figure}
%

\section{Conclusion}
\label{sec:conclusion}
In this paper, we introduce the WiFO system, capable of  improving the wireless capacity of the existing WiFi systems by orders of magnitude.  We describe a cooperative coding schemes called LAC that uses location information to improve the capacity of the receivers in a dense deployment topology.   Both numerical and theoretical results are provided to justify the proposed coding techniques.
\vspace{-.1in}

\medskip

\bibliographystyle{unsrt}
\bibliography{sample}

\end{document}